\documentclass[aps,prd,onecolumn,superscriptaddress,nofootinbib,11.5pt]{revtex4-2}

\usepackage[T1]{fontenc}
\usepackage[utf8]{inputenc}

\usepackage[left=1.05in,right=1.05in,top=1in,bottom=1in]{geometry}

\usepackage{bbm}
\usepackage{amsmath,amssymb}
\usepackage{graphicx}
\usepackage{xcolor}
\usepackage{hyperref}
\usepackage{cleveref}
\usepackage{setspace}
\newcommand{\hl}[1]{#1}

\usepackage{simpler-wick}

\usepackage{feynmp-auto}

\usepackage{amsthm}

\newtheorem{Theorem}{Theorem}

\begin{document}

\title{From quantum time to manifestly covariant QFT: on the need for a quantum-action--based quantization}

\author{N. L. Diaz}
\email{nldiaz.unlp@gmail.com}
\affiliation{Information Sciences, Los Alamos National Laboratory, Los Alamos, NM 87545, USA}
\affiliation{Center for Non-Linear Studies, Los Alamos National Laboratory, Los Alamos, NM 87545, USA}

\begin{abstract}
In quantum time (QT) schemes, time is promoted to a degree of freedom, allowing Lorentz
covariance to be made explicit for single particles. We ask whether this can be lifted to QFT, so that Lorentz covariance becomes manifest at the Hilbert-space level, rather than being hidden as in the standard canonical formulation.
We address this question by proposing a second-quantized approach in which the elementary particle is the QT particle itself, leading naturally to the notion of spacetime field algebras and of quantum action. We show, however, that a naive many-body construction runs into inconsistencies. To pinpoint their origin we introduce a classical counterpart of the second-quantized formalism, spacetime classical mechanics (SCM), and prove a no-go theorem: Dirac quantization of SCM collapses back to standard QFT and therefore hides covariance. We circumvent this problem by presenting a quantum-action--based quantization that yields a spacetime version of quantum mechanics (SQM), making covariance manifest for (interacting) QFTs. Finally, we show that this resolution is tied to a genuine spacetime generalization of the notion of quantum state, required by causality and closely connected to recent ``states over time'' proposals and, in dS/CFT--motivated settings, to microscopic notions of timelike entanglement and emergent time.
\end{abstract}

\maketitle

\section{Introduction}

Time in quantum mechanics (QM) plays a very particular role. 
While other observable quantities are described through quantum mechanical rules, namely, by associated operators on Hilbert space, 
 time is treated as an external  parameter. This indicates a basic asymmetry between space and time, in tension with the formal 
 symmetric role they have in classical physics, and in particular in relativity. On the other hand, contrary to what is usually believed, time can be easily promoted to an observable by following the Page and Wootters mechanism \cite{page1983evolution}, or, in a closely related approach, by promoting time to a dynamical variable in classical phase-space and then quantizing the extended phase-space following Dirac quantization scheme of constrained systems \cite{dirac1950generalized}. A common  feature of these proposals is the use of an enlarged (often called “kinematic”) Hilbert space in which both the system and a time degree of freedom are described jointly. We will refer to such constructions as quantum time (QT) approaches.
While an important motivation for QT formalisms comes from quantum gravity and the problem of time \cite{dewitt1967quantum, unruh1989time, kuchavr2011time, anderson2012problem}, more recently its relevance to quantum information has also been explored \cite{giovannetti2015quantum, boette2016system, boette2018history} providing a renewed perspective that emphasizes the role of quantum correlations between the system and time degrees of freedom. 
These insights have, in turn, been brought to relativistic scenarios \cite{diaz2019historystate, diaz2019history, giovannetti2023geometric, diaz2023spacetime}. In particular, as remarked in \cite{diaz2019historystate, diaz2019history}, in the case of a single quantum particle, QT schemes provide an interesting quantum model where space and time are on an equal footing and Lorentz covariance can be made explicit at the Hilbert space level. 
This claim is  built under a subtle but crucial conceptual shift: In \cite{diaz2019historystate,diaz2019history}
the extended Hilbert space, previously regarded as merely ``kinematic'' (or auxiliary) \cite{marolf1995quantum, hohn2021equivalence}, is treated as the correct arena in which physical states and correlations live, effectively modifying the standard  rules of QM and making the discussion of Lorentz invariance for quantum systems a genuine foundational issue. In other words, the results in \cite{diaz2019historystate,diaz2019history}, and subsequent work \cite{diaz2021spacetime}, pointed out that the proper treatment of relativistic symmetry in QM at the Hilbert space level may require rethinking the axioms of QM itself.

On the other hand, it is widely known that quantum field theory (QFT), the framework underlying the highly successful Standard Model of particle physics, provides a natural setting for relativistic quantum physics. This would seem to belittle the relevance of the previous discussion and of QT schemes altogether: In QFTs space is a classical site-index specifying the label of a particular quantum field at a given time. The fact that both space and time are classical in QFTs suggests that the recent QT discussion is not relevant for QFTs and thus high-energy physics in general. 
However, claiming that space and time are on an equal footing in QFT is highly misleading. Even if both space and time are classical, a deeper space and time asymmetry characterizes QFTs. The problem is not in the QFT approach itself but at the heart of the axioms of QM: Joint systems separated in space are described in a Hilbert space built from the tensor product of the Hilbert spaces of the parts. No such  rule is applied across time. This means that what is meant by space and time needs to be pre-established at the classical level in order to define any quantum theory, including QFTs. 
A direct mathematical consequence is the need to break Lorentz invariance to impose canonical commutation relations as required by Dirac's quantization scheme. 
The obstacle introduced by equal time commutation relations has been extensively discussed in the foundational years of QFTs \cite{dyson1949radiation} 
and it was concluded that Lorentz symmetry, while hidden, is not truly lost as the physical predictions of relativistic QFTs are perfectly compatible with relativity. 
This statement is also compatible with the Path Integral (PI) approach \cite{feynman1948space, feynman1949space} that allows one to preserve Lorentz symmetry explicitly at the price of abandoning a Hilbert space structure altogether.

The previous considerations bring us back to the foundational discussion about the dichotomy between quantum mechanics and relativity. 
Given that Lorentz invariance is ultimately present in relativistic QFTs, even if hidden by the standard canonical formalism, 
we may pose a natural but somehow overlooked question: Is it possible to reformulate (or minimally extend) the canonical approach so that Lorentz symmetry is manifest in QFTs at the Hilbert space level (rather than recovering it only at the level of predictions)? Can the lessons from single particle (sp) QT models  be lifted to QFTs? 
Our answer is yes, but it requires abandoning Dirac-style constraint quantization in the many-body case and replacing it with a ``quantum-action-based'' 
quantization that implements the constraints inside correlators. This is the main result of the current manuscript. In order to prove it we will make use of the recently proposed framework of  \emph{spacetime quantum mechanics} (SQM) \cite{diaz2023spacetime,diaz2025spacetime} able to tackle the problem of imposing tensor products across timelike separations we discussed above. 
It is worth noting that the SQM approach is particularly adequate to describe quantum fields \cite{diaz2023spacetime} and while rooted in a Hilbert space approach it also provides a novel perspective on the PI formulation \cite{diaz2021path} that can be extended to fermions \cite{diaz2025spacetime}. Moreover, a spacetime classical mechanics (SCM) approach can also be formulated accordingly \cite{diaz2023spacetime}. In this work the SQM approach will be rederived self-consistently from our considerations on QT, showing that it can be reinterpreted as the proper (but subtle) many-body completion of QT schemes. Considering that several mathematical subtleties and details of the results presented here have been tackled in the previous SQM literature, we will refer to these discussions when appropriate.

The manuscript is organized as follows. In Sec.~\ref{sec:sp} we introduce the QT scheme for a single particle via Dirac’s quantization of a classical particle formulated in an extended phase space containing time, and we recall how this standard construction leads to a “universe equation” \cite{kiefer2012quantum}. In Sec.~\ref{sec:secondq} we propose a naive second quantization of this formalism; while this step already contains many of the ingredients needed for an explicitly covariant QFT, including a quantum version of the classical action, it does not yet provide a consistent many-body scheme, and in Sec.~\ref{sec:limitations} we show how unavoidable inconsistencies arise when one attempts to generalize the universe equation to many particles. In Sec.~\ref{sec:SCM} we introduce the SCM approach as a classical counterpart of the structure suggested by the second-quantized QT construction, which allows us to identify the proper constraints and pinpoint the origin of the previous inconsistencies. In Sec.~\ref{sec:DiracquantizationSCM} we discuss Dirac’s quantization of SCM and prove a no-go theorem showing that Dirac quantization of SCM leads back to standard QFT. In Sec.~\ref{sec:quantumactionquantization} we present an alternative, quantum-action-based quantization and show that, when applied to SCM, it leads to SQM instead; in particular, we show that the ability to circumvent the use of Dirac brackets before quantization is precisely what allows one to preserve manifest Lorentz symmetry in the quantum setting while (i) avoiding the inconsistencies of a direct second-quantized-QT construction and (ii) preserving the QT notion of particle. In Sec.~\ref{sec:interactions} we explain how the formalism naturally applies to interacting QFTs. 
Finally, in Sec.~\ref{sec:ststates} we explain how the need for a non-standard quantization scheme is linked to the need for a non-standard generalization of the notion of quantum state to spacetime; this connects our considerations to recent discussions on defining quantum states in spacetime \cite{horsman2017can, fitzsimons2015quantum, cotler2018superdensity, fullwood2024operator, diaz2025spacetime, milekhin2025observable, guo2025spacetime}, and supports the conclusion that genuinely non-trivial generalizations of the standard notion of state are required.

\section{Classical and quantum dynamical time}

\subsection{Single Particle: From time in phase-space to quantum time via Dirac quantization}\label{sec:sp}

Let us briefly recall a procedure due to Dirac \cite{dirac1950generalized} that allows one to include time in phase space. 
For simplicity, consider a one-dimensional particle (the general case is entirely analogous) described by a time-independent Lagrangian $L(q,\dot q)$,
\begin{equation}
S[q(t)] = \int_{t_1}^{t_2} dt\, L(q,\dot q).
\label{eq:B1}
\end{equation}
We may promote $t$ to a coordinate in configuration space by 
introducing a parameter $\tau$ that labels the trajectory $(t(\tau),q(\tau))$. Then the action can be rewritten as
\begin{equation}
S[q(\tau),t(\tau)] 
= \int_{\tau_1}^{\tau_2} d\tau\, \dot t\, L\!\left(q,\frac{\dot q}{\dot t}\right)
\equiv \int_{\tau_1}^{\tau_2} d\tau\, \tilde L(q,\dot q,\dot t).
,
\label{eq:B2}
\end{equation}
for $\tau_i=\tau(t_i)$.
The canonical momenta associated with $\tilde L$ are (see, e.g., standard treatments of reparametrization-invariant systems \cite{kiefer2012quantum} for details)
\begin{equation}
\tilde p_q=\frac{\partial \tilde L}{\partial \dot q}=p_q,
\qquad
p_t=\frac{\partial \tilde L}{\partial \dot t}=-H,
\label{eq:B3}
\end{equation}
where $H$ is the usual Hamiltonian obtained from $L(q,\dot q)$. 
The corresponding Hamiltonian for the parameterized system  is instead given by
\begin{equation}
\tilde H=\tilde p_q \dot q + p_t \dot t - \tilde L
= \dot t\,(H+p_t).
\end{equation}
Considering that the equations of motion yield $p_t=-H$ we must impose the 
(weak) constraint
\begin{equation}
H_S \equiv -(p_t+H) \approx 0,
\label{eq:B4}
\end{equation}
where $\approx$ denotes a constraint to be imposed after all brackets have been evaluated. Here $H_S$ denotes what is usually called a ``super Hamiltonian'' (defined with an overall minus sign for later convenience) and the equation \eqref{eq:B4} might be denoted as a Wheeler-de-Witt like equation, since a similar constraint appears in a canonical quantization of gravity \cite{dewitt1967quantum} (therein parameterization invariance is inherited from a genuine covariance under general diffeomorphisms).

Let us now focus on the relativistic action of a scalar particle. In this context starting with an exhibiting parameterization invariance is natural: One can take $L=-m\sqrt{\dot{x}^\mu \dot{x}_\mu}$ where $x^\mu(\tau)\equiv (t(\tau),q(\tau))$ so the action is invariant under $\tau\to \tau'(\tau)$. In particular, the gauge choice $t=\tau$ leads to $S=-m\int dt\, \sqrt{1-\dot{q}(t)^2}$ and to $p(\tau)=\frac{m \dot{q}}{\sqrt{1-{\dot{q}^2}}}$, $H=\sqrt{p^2+m^2}$. If we start with the latter form of the action is clear that Eq.\ \eqref{eq:B4}  now reads $p_t+\sqrt{p^2+m^2}\approx 0$. Another possibility is to work directly with a general $\tau$ leading instead to the constraint $p^\mu p_\mu-m^2\approx 0$.  For our purposes the form of Eq.\ \eqref{eq:B4} suffices.

The next step in Dirac's quantization corresponds to promote the Poisson brackets (PBs) $\{t,p_t\}=1$, $\{q,p_q\}=1$ to commutators  $[t,p_t]=i$, $[q,p_q]=i$ (we set $\hbar=1$). Since all other commutators vanish (e.g. $[t,q]=[t,p_q]=0$) we end up in a Hilbert space with the structure $h=h_t\otimes h_q$. Then, we impose all constraints from the classical theory by defining physical subspaces. In our present case there is only one constraint so we just need $H_S|\Psi\rangle=0$. A direct consequence of this constraint is that within the physical subspace defined by the previous equation $\langle p_t\rangle= -\langle H\rangle$, namely expectation values computed with either operator give the same result. Since $p_t$ generates time translations in the time sector ($e^{i\tau p_t}|t\rangle=|t+\tau\rangle$ for $\hat{t}|t\rangle =t |\rangle$), the constraint is thus identifying time translations in time (in the geometrical sense) with the flow generated by the Hamiltonian.  Notice that here we are considering an expectation value computed directly within $h$. Under the standard Dirac interpretation the Kinematic space $h$ should instead be abandoned, and the physical subspace is not to be regarded as an actual subspace but as a tool to induce a physical inner product within a new Hilbert space \cite{marolf1995quantum}. Crucially, we won't follow Dirac's interpretation and work directly within $h$ so the physical subspace is for us a true subspace and we preserve the time operator (which under Dirac's scheme is no longer present within the ``actual'' Hilbert space).

The previous interpretation is more in line with the PaW mechanism where the Hilbert space $h$ is regarded as physical since $h_T$ is, following PaW, to be considered as the Hilbert space of an actual physical system (sometimes referred to as the clock \cite{giovannetti2015quantum}). Within the PaW scheme standard expectation values at a given time are recovered by conditioning on ``clock readings'': By defining \begin{equation}\label{eq:Ot}
    O(t):=|t\rangle \langle t|\otimes O
\end{equation} one obtains
\begin{equation}\label{eq:conditioning}
    \langle \psi(t)|O|\psi(t)\rangle=\langle \Psi|O(t) |\Psi\rangle\,,
\end{equation}
where $|\Psi\rangle$ satisfies the universe equation $H_S|\Psi\rangle
=\bigl(P_t\otimes I + I\otimes H\bigr)|\Psi\rangle =0$ that leads to $|\Psi\rangle=\int dt\, |t\rangle |\psi(t)\rangle$ for $|\psi(t)\rangle=e^{-iHt}|\psi\rangle$ (with $|\psi\rangle$ the initial state). This also means that
\begin{equation}
    \langle \Psi|(i\dot{O}(t)+[H,O(t)])|\Psi\rangle=0\,,
\end{equation}
i.e., the operator $O(t)$ satisfies a Heisenberg-like equation, but with $\dot{O}(t)\equiv \lim_{\epsilon\to 0}[O(t+\epsilon)-O(t)]/\epsilon$ a ``geometrical'' derivative (we recall that here $t$ labels states, not unitary evolution). 
While we won't follow the PaW interpretation regarding conditioning on a physical system, we will discuss the value of equation \eqref{eq:conditioning} within a many-body scenario. Let us notice that Eq.\ \eqref{eq:conditioning} is an ``operator'' version of the standard conditioning where one applies the projector in time $|t\rangle \langle t|$
directly to $|\Psi\rangle$ instead \cite{giovannetti2015quantum, boette2016system}. While at this sp level the two approaches are equivalent, when generalizing to many particle  we will find more useful to consider the operator conditioning scheme.

Notice also that according to Dirac's scheme $O(t)$ is not a physical operator meaning that it is not defined once the kinematic space is abandoned. It is also worth mentioning that working directly within $h$ may seem to cause problems regarding the inner product of physical states, with infinities easily arising. While these infinities can be regularized by considering e.g. finite time windows of discrete time, in the case of a single particle  a straightforward interpretation is possible: Within $h$ the universe equation can be interpreted as the eigenvalue equation of a mass operator. Then, a Dirac delta like normalization is required since mass can take continuum values. This provides a consistent way to treat these infinities that has a proper non-relativistic limit. While these subtleties are not particularly relevant for our purposes, and we refer the reader to  \cite{diaz2019historystate, diaz2019history} for details, let us add that the mechanism of inducing a physical inner product, in essence, is just avoiding redundant projections onto physical subspaces \footnote{If $\Pi_m\equiv |m\rangle\langle m|$ projects onto the physical subspace $\langle \Phi|\Pi_{m}\Pi_{m'}|\Psi\rangle=\delta(m-m')\langle \Phi|\Pi_{m}|\Psi\rangle$ so one might replace $\langle \Phi|\Pi_{m}\Pi_{m'}|\Psi\rangle\to \langle \Phi|\Pi_{m}|\Psi\rangle$ to get a finite result; The same result can be obtained if a small mass uncertainty is allowed; we refer the reader to \cite{diaz2019history} for details.}.

Let us now focus entirely on the case of the Klein Gordon particle in arbitrary $D=d+1$ dimensions. In the extended Hilbert space $h$ we have $[X^\mu,P_\nu]=i\delta^{\mu}_{\;\,\nu}$. Here $X^0$ is the time operator and $h$ has the Hilbert space structure $h=h_T \otimes h_S$ for $h_S=\text{span}\{|\textbf{x}\rangle\}$ the standard Hilbert space of a $d$-dimensional particle. A complete basis of $h$ is given by eigenstates of $X^\mu$ that we denote as $|x\rangle=|x_0\rangle \otimes |x_1\rangle \otimes \dots |x_d\rangle$ such that $\hat{X}^\mu |x\rangle= x^\mu |x\rangle$. Similarly we write $\hat{P}^\mu |p\rangle= -p^\mu |p\rangle$, equivalent to $\hat{P}_0|p\rangle=-p_0|p\rangle$ and $\hat{P}_i|p\rangle=p_i|p\rangle$ where we adopt the mostly minus convention for the metric $\eta_{\mu \nu}=\text{diag}(1,-1,-1,\dots, -1)$. With these conventions $\hat{P}^0=\hat{P}_0=\hat{p}_t$, while we fix the momentum-eigenstate convention so that the overlap
 $\langle p|x\rangle=e^{ipx}$ is a Lorentz scalar, as follows from $px\equiv p^\mu x_\mu=p^0 x^0-\textbf{p}.\textbf{x}$.
These bases also satisfy
\begin{equation}\label{eq:orthogon}
    \langle x|y\rangle=\delta^{(D)}(x-y)\,,\;\; \langle p|q\rangle=(2\pi)^D\delta^{(D)}(p-q)\,,
\end{equation}
which are just the standard canonical relations now extended to time.

Notice that while so far we haven't imposed any physical condition space and time are clearly on an equal footing. We can define unitary Lorentz transformations accordingly: $U(\Lambda):=\exp[i \omega^{\mu\nu}L_{\mu\nu}]$ induces $U(\Lambda )|x\rangle=|\Lambda x\rangle$ and $U(\Lambda )|p\rangle=|\Lambda p\rangle$ for $\Lambda=e^\omega$ and $L_{\mu\nu}=X_\mu P_\nu-X_\nu P_\mu$ the angular momentum in spacetime. The operators $L_{\mu\nu}$, and $P_\mu$ close the Poincar\'e algebra. Interestingly, we see that $h$ admits an explicit representation of special relativistic symmetries that does not involve dynamical information. In other words, we don't need to introduce a Hamiltonian to define $U(\Lambda)$, and the proper Poincar\'e algebra is recovered since $P^0$ is also independent of the Hamiltonian, and in particular $[P_\mu,P_\nu]=0$.

On the other hand, in order to introduce dynamics we recall that the universe equation reads $(P^0+\sqrt{\textbf{P}^2+m^2})|\Psi\rangle=0$, which is equivalent to $P^2|\Psi\rangle=m^2|\Psi\rangle$ with the additional condition of positive $P^0$. By projecting onto the position eigenbasis we thus recover Klein-Gordon equation:
\begin{equation}
    \langle x| (P^2-m^2)|\Psi\rangle=-(\partial^2+m^2)\Psi(X)=0\,,
\end{equation}
for $|\Psi\rangle=\int d^Dx \, \Psi(x)|x\rangle$ and where we used that $\langle x|P_\mu |y\rangle=-i\partial_\mu \langle x|y\rangle$.

A positive energy solution has then the form $|\Psi\rangle_m=\int \frac{d^Dp}{(2\pi)^D}\theta(P^0)2\pi\delta(P^2-m^2)\psi(\textbf{p})|p\rangle$ which is equal to $\int \frac{d^dp}{(2\pi)^d}\psi(\textbf{p})|E_{\textbf{p}m},\textbf{p}\rangle$ with $\theta$ the Heaviside step function and $E_{\textbf{p}m}=\sqrt{\textbf{p}^2+m^2}$. In the time basis we can also write
\begin{equation}\label{eq:physm}
    |\Psi\rangle_m=\int dt\int \frac{d^dp}{(2\pi)^d}e^{-iE_{\textbf{p}m}t}\psi(\textbf{p})|t,\textbf{p}\rangle\,,
\end{equation}
with the same expression for negative-energy solutions after the replacement $E_{\textbf{p}m}\to -E_{\textbf{p}m}$. We mention that these states can be regarded as mass eigenstates such that their orthonormality under the $D$-dimensional Schr\"odinger's inner product leads to $\psi(\textbf{p})$ normalized according to the Klein-Gordon norm \cite{diaz2019history}.

Let us now notice that $[U(\Lambda), P^2]=0$ indicates Lorentz symmetry at the dynamical level. We remark that in this context $U(\Lambda)$ is well-defined even if a given universe operator doesn't commute with it. The representation of the symmetry has been decoupled from the dynamics. This happens with standard quantum mechanical representation of symmetries: For example, 
a Hilbert space can provide a representation of the rotation group even if the Hamiltonian of the system in question doesn't exhibit rotational symmetry. 
The QT framework is thus particularly aligned with the standard spirit of QM, provided we preserve the full $h$.

\subsection{Second quantization: From the universe operator to the quantum action operator}\label{sec:secondq}

\begin{figure}[t!]
    \centering
\includegraphics[width=0.63\linewidth]{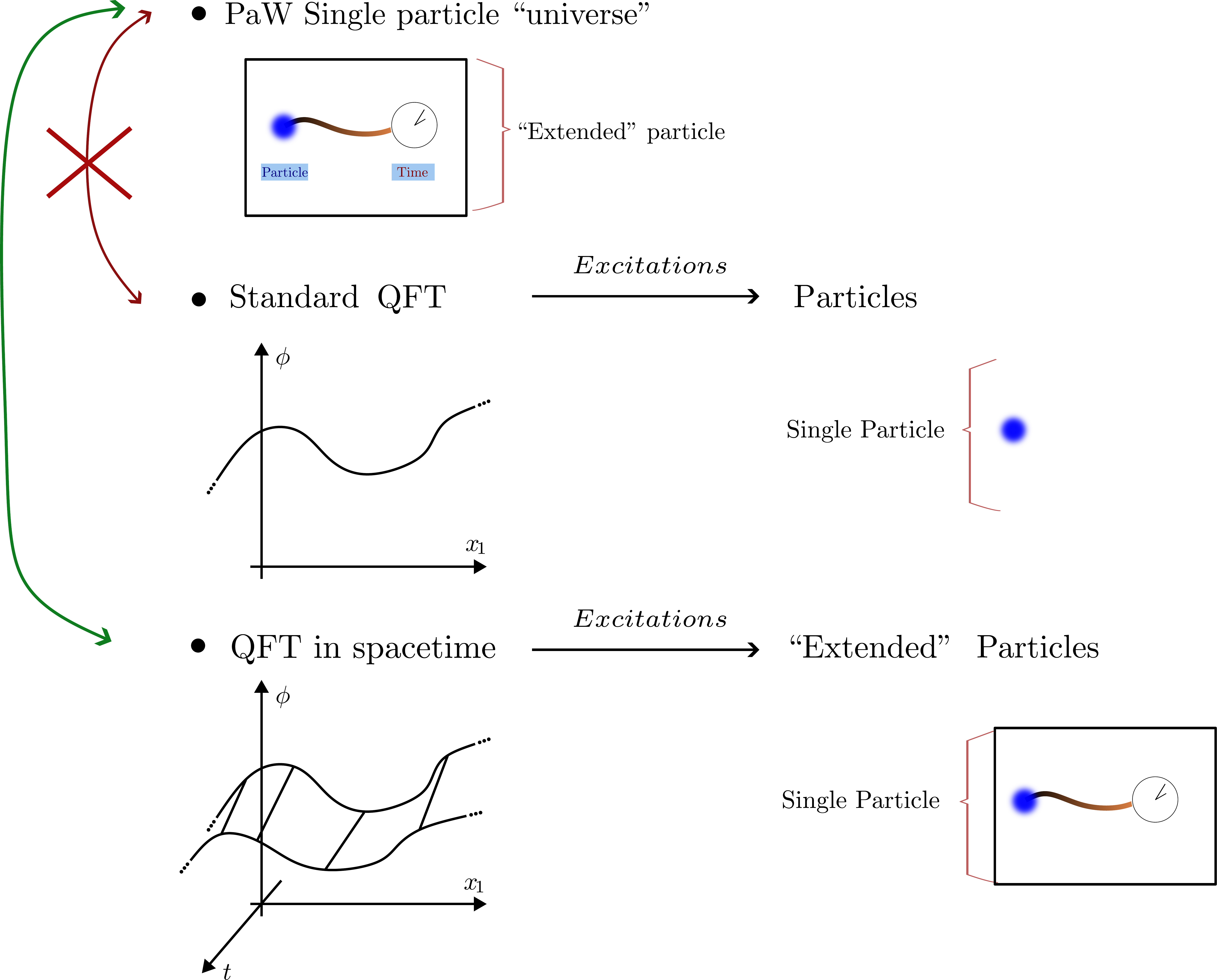}
    \caption{Scheme on the relation between the QT scheme of a particle, the standard notion of particle in QFT and the connection between the two developed here. On top we depict a single particle formulated according to the PaW mechanism with time as a degree of freedom. This defines an extended notion of quantum particle. On the other hand, according to QFTs, a particle is a quantum excitation of an underlying field. However,
    the particles corresponding to standard QFT do not include the time degree of freedom. In this manuscript we describe a formulation of QFT that is directly connected to the PaW notion of particle through second quantization: Contrary to standard QFT, where equal-time commutators are imposed, such that the Hilbert space describes field configurations at a given time, the extended formulation assumes a canonical algebra that treats space and time indices on an equal-footing. The ensuing Hilbert space admits a basis of field configurations in spacetime and is also isomorphic to a Fock space where the single particle states contain time as a degree of freedom.}
    \label{fig:fig1}
\end{figure}

It is natural to discuss the relevance of the previous construction in a many-body setting. Let us first notice that the QT notion of particle cannot be extended by a direct application of the QT scheme to a Fock space: Schematically, the Hilbert space structure $h_T\otimes h_{\text{Fock}}$ contains sp states with the previous form, but also two particle states like $|t\rangle \otimes a^\dag(\textbf{p}_1)a^\dag(\textbf{p}_2)|0\rangle$, and so on, for which a proper Lorentz transformation is clearly out of reach. In this context $t$ could represent a proper time, but not the time corresponding to a given reference frame.

Instead, we propose to construct a new Fock space where the elementary sp is taken directly as the QT-particle. Mathematically, we can then define the Fock space $\mathcal{H}=\text{Sym}[\oplus_n h^{\otimes n}]$ for $h$ the QT Hilbert space with spacetime basis $|x\rangle$.  Let us remark that we are preserving the whole $h$ space in this construction. We also recall that in this context ``Sym'' indicates that we symmetrize when multiple copies of the system are present, as we should for fundamental indistinguishable particles (for fermions one just replaces with an antisymmetrization). For example a two particle state may be expanded in the basis
\begin{equation}
\frac{|x_1\rangle|x_2\rangle+|x_2\rangle|x_1\rangle}{\sqrt{2}}\equiv    \frac{|t_1,\textbf{x}_1\rangle\otimes  |t_2,\textbf{x}_2\rangle+|t_2,\textbf{x}_2\rangle\otimes  |t_1,\textbf{x}_1\rangle}{\sqrt{2}}\,.
\end{equation}
Interestingly, the structure ``time'' $\otimes$ ``space'' that was present for a single particle, doesn't hold for many-particle states. Instead, a ``multi-time'' like structure appears. 
In order to better understand the structure of $\mathcal{H}$ let us recall that a far more convenient description of a Fock space is obtained by defining ladder operators: In our case this leads one to $|x\rangle=a^{\dag}(x)|\Omega\rangle$ for $|\Omega\rangle$ the vacuum state annihilated by all $a(x)$ and to
\begin{equation}\label{eq:basicalg}
    [a(x),a^\dag(y)]=\delta^{(D)}(x-y)\,,
\end{equation}
in agreement with Eq.\ \eqref{eq:orthogon}. Equation \eqref{eq:basicalg} is of fundamental importance: This is a spacetime algebra treating space and time as indices. As a matter of fact, at this level space and time are not distinguishable with \eqref{eq:basicalg} isomorphic to the standard canonical algebra in $d+1$ (spatial) dimensions. On the other hand, this scheme is in clear contrast to the standard QFT quantization where canonical commutators require one to consider equal times. Moreover, if we discretize spacetime so that 
$[a_{t,\textbf{x}},a^\dag_{t',\textbf{x}'}]=\delta_{tt'}\delta_{\textbf{x},\textbf{y}}$ ($a(x)\equiv a_{t,\textbf{x}}/\sqrt{\epsilon^D}$ for a uniform spacing $\epsilon$ on all dimensions; see details on the continuum spacetime limit on Appendix \ref{app:continuumtimelimit}), where the deltas are now Kronecker's delta, it becomes clear that the tensor product structure that characterizes quantum systems separated in space has been extended to time. One way of seeing this is to notice that $a_{t,\textbf{x}}\equiv \mathbbm{1}\otimes \mathbbm{1}\dots a_{\textbf{x}}\otimes \mathbbm{1}\dots$, namely the operator $a_{\textbf{x}}$ acting on the time-slice $t$. This is precisely the structure that allows one to define standard entanglement in continuum variables systems \cite{eisert2010colloquium}. Hence, at least for discrete spacetime we can write $\mathcal{H}\simeq \otimes_{t}\otimes_{\textbf{x}}\mathcal{H}^{\text{single mode}}_{t,\textbf{x}}$, for $\mathcal{H}^{\text{single mode}}$ the Hilbert space corresponding to a single bosonic mode defined by $[a,a^\dag]=1$. Notably, the Fock space where $h$ is the building object applies a tensor product structure both across space and time. Moreover, $\mathcal{H}\simeq \otimes_{t}[\otimes_{\textbf{x}}\mathcal{H}^{\text{single mode}}_{t,\textbf{x}}]\simeq\mathcal{H}_{\text{space}}^{\otimes N}$ (for $N$ time slices), namely the extended Fock space is isomorphic to the Hilbert space obtained by considering the standard Hilbert space $\mathcal{H}_{\text{space}}$ defined by equal-times commutator relations and copying it across time slices.

Notice now that while we used the ``position'' basis we can also define ladder operators $a(p)$ such that $|p\rangle\equiv a^\dag(p)|\Omega\rangle$ with $|p\rangle$ defined in the previous section. The relation between the position and momentum ladder operators is given by the spacetime Fourier transform
\begin{equation}\label{eq:axap}
    a(x)=\int \frac{d^Dp}{(2\pi)^D}e^{-ipx}a(p)\,,
\end{equation}
One may also consider independent Fourier transforms for each dimension to define e.g. $a(t,\textbf{p})$. Notice that since the operators $a(x)$ and $a(p)$ are linearly related and no creation operator is involved in Eq.\ \eqref{eq:axap} they share the same vacuum $|\Omega\rangle$ (in other words, they are related by a vacuum preserving Bogoliubov transformation; this needs not to hold when considering ladder operators defined as normal modes of a field Hamiltonian, see below). Remarkably, since the measures $d^Dp, d^Dx$ are clearly invariant under Lorentz transformations, and so are the product $px=p^\mu x_\mu$ and the $D$-dimensional Dirac delta, we can represent Lorentz transformations explicitly: If we define $L_{\mu\nu}:=\int d^Dx\, a^\dag(x) (x_\mu\partial_\mu-x_\nu \partial_\mu)a(x)=\int \frac{d^Dp}{(2\pi)^D}\, a^\dag(p) (p_\mu\partial_\mu-p_\nu \partial_\mu)a(p)$ then
\begin{equation}
U(\Lambda):=\exp\left(i\omega_{\mu\nu}L^{\mu\nu}\right)
\end{equation}
implements Lorentz transformations under adjoint action, i.e. $U^\dag(\Lambda)a(x)U(\Lambda)=a(\Lambda x)$ and   $U^\dag(\Lambda)a(p)U(\Lambda)=a(\Lambda p)$, for $\Lambda=e^\omega$. Clearly, $U(\Lambda)$ is unitary and preserves the spacetime algebra \ref{eq:basicalg} as well as 
\begin{equation}\label{eq:momentumalg}
    [a(p),a^\dag(k)]=(2\pi)^D\delta^{(D)}(p-k)\,.
\end{equation}

Having identified the structure underlying the extended second quantization, let us now discuss how QT operators are generalized to this setting. Fortunately, the second quantization procedure provides us with a rigorous and practical recipe: Given a general sp operator
$O=\sum_{i,j}\langle i|O|j\rangle|i\rangle \langle j|$ the corresponding Fock space operator, defined as the sum of $O$ operators corresponding to  different particles
(e.g. the total kinetic energy), is given by 
$
    O_F:=\sum_{i,j} \,\langle i|O|j\rangle \,a_i^\dag a_j\,.
$
We refer the reader to \cite{schwabl2008advanced} for a rigorous and particularly clear discussion on second quantization from a modern point of view. We can apply the same procedure to operators $O(t)$ defined in Eq. \eqref{eq:Ot} which for a Klein-Gordon particle reads $O(t)=\int \frac{d^dp}{(2\pi)^d} \frac{d^dp'}{(2\pi)^d}  O(\textbf{p},\textbf{p}')|t,\textbf{p}\rangle \langle t,\textbf{p}'|$.  As a result one obtains,
\begin{equation}
   O(t)\to O_F(t)=\int \frac{d^dp}{(2\pi)^d} \frac{d^dp'}{(2\pi)^d} O(\textbf{p},\textbf{p}')a^\dag(t,\textbf{p})a(t,\textbf{p}')\,.
\end{equation}
Interestingly, standard sp operators that in the PaW approach implement time conditioning at time $t$ get generalized to local-in-time operators. In fact, ignoring subtleties of the continuum (see comments above), we can state that $O_F(t)$ is indeed acting on the copy $t$ of the standard Fock space defined by ladder operators $a(\textbf{p}), a^\dag(\textbf{p})$. This notion of locality of operators is very relevant for the discussion of section \ref{sec:SCM}.

Just as we applied the second quantization recipe to $O(t)$ we may apply it to another fundamental sp operator: the universe operator $H_S$. By expanding $H_S$ in the $(t,\textbf{p})$ basis we can write
\begin{equation}\label{eq:freekgaction}
\begin{split}
      H_S&=-\int dt \int \frac{d^dp}{(2\pi)^d}|t,\textbf{p}\rangle [-i\partial_t +\sqrt{\textbf{p}^2+m^2}] \langle t,\textbf{p}|\\&\to \mathcal{S}:= \int dt \int \frac{d^dp}{(2\pi)^d}\left[a^\dag(t,\textbf{p}) i\dot{a}(t,\textbf{p}) -\sqrt{\textbf{p}^2+m^2}\,a^\dag(t,\textbf{p}) a(t,\textbf{p})\right]\,,
\end{split}
\end{equation}
with $\mathcal{S}\equiv (H_S)_F$, i.e. the Fock space version of $H_S$. The careful reader may  have noticed a rather peculiar structure in $\mathcal{S}$: The second term is formally $-\int dt\, H$ with $H$ the integral over time of the quantum Hamiltonian corresponding to a free Klein-Gordon field. Moreover, the first term has a Legendre transform structure with $ia^\dag$ the conjugate variable of $a$. In other words, $\mathcal{S}$ has precisely the structure of an \emph{action}. To make further progress let us define field operators  
\begin{equation}\label{eq:fields}
    \begin{split}
        \phi(x):&=\int \frac{d^Dp}{(2\pi)^D}\frac{1}{\sqrt{2 E_\textbf{p}}}\Big(a(p)e^{-ipx}+\text{h.c}\Big)\\
        \pi(x):&=-i\int \frac{d^Dp}{(2\pi)^D}\sqrt{\frac{ E_\textbf{p}}{2}}\Big(a(p)e^{-ipx}-\text{h.c}\Big)
    \end{split}\,.
\end{equation}
It is a straightforward exercise to invert these relations and  rewrite
    \begin{equation}
    \mathcal{S}=\int d^Dx \, \left(\pi(x)\dot{\phi}(x)-\frac{\pi^2(x)}{2}-\frac{(\nabla \phi(x))^2}{2}-\frac{m^2\phi^2(x)}{2} \right)\,.
\end{equation}
We reach the following remarkable conclusion: \emph{The many-particle generalization of the universe operator of a Klein-Gordon particle is a quantum version of the Klein-Gordon field free action}. This result has also been presented recently for Dirac particles/fields \cite{diaz2025spacetime} and holds for non-relativistic systems as well \cite{diaz2021spacetime}.

Let us make a few conceptual comments about the previous important result. First of all, let us emphasize that $\mathcal{S}$ is quantum operator, which for obvious reasons we denote as \emph{quantum action}. Secondly, we stress that $\mathcal{S}$ is only a meaningful object in this extended setting: Since
\begin{equation}\label{eq:stfieldalg}
    [\phi(x),\pi(y)]=i\delta^{(D)}(x-y)\,,
\end{equation}
as one can easily check using the algebra of $a(p)$ \footnote{As a technical remark, we clarify that the relation between $a(x)$ in Eq.\ \eqref{eq:axap} and the field operators of Eq.\ \eqref{eq:fields} involves a non-trivial Bogoliubov transformation.}, fields at different time points are independent and the structure of $\mathcal{S}$ cannot be simplified. In a sense, this mimics the situation in classical physics where the action is defined by generic fields without assuming any equation of motion as they are instead derived from the formal structure of the action.  Moreover, the ladder operators that appear in the field expansions \eqref{eq:fields} are the one creating/annihilating PaW-like particles. We have thus successfully identified the proper QFT scenario  that reinterprets PaW-like particles as field excitations. At the same time, while in standard QFT one may formally think of the Hilbert space as the space of field configurations at a given time, with an eigenbasis $|\phi(\textbf{x})\rangle$ of operators $\phi(\textbf{x})$, we can think of $\mathcal{H}$ as the space of field configurations in spacetime, with eigenbasis  $|\phi(x)\rangle$ such that $\hat{\phi}(x)|\phi(x)\rangle= \phi(x)|\phi(x)\rangle$ which follows from Eq.\ \eqref{eq:stfieldalg} (up to the standard continuum subtleties). 
These comments are illustrated in Figure \ref{fig:fig1}.

\subsection{Limitations of a constrained equation in the extended Fock space}\label{sec:limitations}

Now that we have generalized the main QT operators to a many-particle scenario, we discuss the possibility of extending the notions of physical subspace and conditioning to this setting. As we discussed in section \ref{sec:sp}, these are the basic ingredients that in the PaW approach allow one to recover standard unitary evolution. Since it is not a priori clear how to treat evolution in a Hilbert space applying a tensor product structure across time, it is natural to leverage what we know from QT schemes.

First of all, let us notice that the sp physical state of Eq.\ \eqref{eq:physm} now can be written as
\begin{equation}
 |\Psi\rangle_m=  \int \frac{d^dp}{(2\pi)^d}\psi(\textbf{p})a^\dag(E_{\textbf{p}m},\textbf{p})|\Omega\rangle= \int dt\int \frac{d^dp}{(2\pi)^d}e^{-iE_{\textbf{p}m}t}\psi(\textbf{p})a^\dag(t,\textbf{p})|\Omega\rangle\,.
\end{equation}
Interestingly, the creation operator $a^\dag(E_{\textbf{p}m},\textbf{p})$ corresponds to an on-shell particle. This suggests to define the physical subspace as the space generated by on-shell modes such that, e.g. a two particle state looks like
\begin{equation}
\begin{split}
     |\Psi^{(2)}\rangle_m&=  \int \frac{d^dp}{(2\pi)^d}\frac{d^dp'}{(2\pi)^d}\psi(\textbf{p},\textbf{p}')a^\dag(E_{\textbf{p}m},\textbf{p})a^\dag(E_{\textbf{p}'m},\textbf{p}')|\Omega\rangle\\&= \int dtdt' \int \frac{d^dp}{(2\pi)^d}\frac{d^dp'}{(2\pi)^d}\psi(\textbf{p},\textbf{p}')e^{-i(E_{\textbf{p}m}t+E_{\textbf{p}'}mt')}a^\dag(t,\textbf{p})a^\dag(t',\textbf{p}')|\Omega\rangle\,.
\end{split}
\end{equation}
We see that a multi-time-like structure, not present for sp QT states, emerges. At the same time, the structure of multi-particle states is simple to understand in the mass-energy basis where a correspondence to standard QFT states is manifest as long as the on-shell condition is imposed strictly (namely, the wavefunctions $\psi$ only depend on the spatial part). The condition on physical states can be stated as follows: Physical modes are defined by the condition 
\begin{equation}\label{eq:scommphys}
    [\mathcal{S},a^\dag(p)]=0\,,
\end{equation}
which yields $[\mathcal{S},a^\dag(p)]=(p^0-E_{\textbf
{p}m})a^\dag(p)=0$ holding only for $p^0=E_{\textbf{p}m}$. Then, the physical subspace is the linear space of states built from the action of physical creation operator on the vacuum. In particular, this implies $\mathcal{S}|\Psi\rangle=0$ thus generalizing the universe operator condition for single particles. We also remark that the physical condition gives the quantum action a central role.

Let us now explore the consequences of the previous construction. Consider the expectation values of an operator $O_F(t)$. By the second quantization construction, ${}_m\langle \Psi| O_F(t)|\Psi\rangle_m$ is equal to the sp mean value which, at the same time, is equal to $\langle \psi(t)|O|\psi(t)\rangle$ as we recalled in Eq.\ \eqref{eq:conditioning}.  Let us verify this explicitly
\begin{equation}\label{eq:contractionOt}
\begin{split}
      {}_m\langle \Psi| O_F(t)|\Psi\rangle_m=\int &\frac{d^dp}{(2\pi)^d} \frac{d^dp'}{(2\pi)^d}  \frac{d^dk}{(2\pi)^d} \frac{d^dk'}{(2\pi)^d} \psi^\ast(\textbf{p}') O(\textbf{k}',\textbf{k})\psi(\textbf{p})\times \\&\;\langle \Omega| \wick{\c1 a(E_{\textbf{p}'m},\textbf{p}') \c1 a^\dag(t,\textbf{k}') \c2 a(t,\textbf{k}) \c2 a^\dag(E_{\textbf{p}m},\textbf{p})}|\Omega\rangle\\
      =\int &\frac{d^dp}{(2\pi)^d} \frac{d^dp'}{(2\pi)^d}  \psi^\ast(\textbf{p}',t) O(\textbf{p}',\textbf{p})\psi(\textbf{p},t)\,,
\end{split}
\end{equation}
which is equal to $\langle \psi(t)|O|\psi(t)\rangle$ for $|\psi(t)\rangle=\int \frac{d^dp}{(2\pi)^d}\psi(\textbf{p})e^{-iE_{\textbf{p}m}t}a^\dag(\textbf{p})|0\rangle$ ($\psi(\textbf{p},t)=\psi(\textbf{p})e^{-iE_{\textbf{p}m}t}$) which is a sp state evolved by an external time and corresponding to standard (non-extended) ladder operators defined by $[a(\textbf{p}),a^\dag(\textbf{k})]=(2\pi)^d \delta^{(d)}(\textbf{p}-\textbf{k})$ (with similar considerations for $O$ and with $a(\textbf{p})|0\rangle=0$). Here the contractions are defined as $\wick{\c1 A \c1 B}:=\langle \Omega|AB|\Omega\rangle$ and the fact that $|\Omega\rangle$ is gaussian allows one to apply Wick's theorem (no time-ordering is involved here). We notice that all internal contractions within $O_F(t)$ vanish ($\wick{\c1 a^\dag(t,\textbf{k}') \c1 a(t,\textbf{k})}=\langle \Omega|a^\dag(t,\textbf{k}')  a(t,\textbf{k}) |\Omega\rangle=0$), so that  the only relevant is the one considered in Eq.\ \eqref{eq:contractionOt}.
We see that the mean value computed within $\mathcal{H}$ is equal to the corresponding mean value computed within standard QFT. This is a direct generalization of the PaW mechanism to Fock space.

Interestingly, if we consider a two particle mean value the sp result gets immediately generalized $
    {}_m\langle \Psi^{(2)}|O_F(t)|\Psi^{(2)}\rangle_m=\langle \psi^{(2)}(t)|O|\psi^{(2)}(t)\rangle
$
for $|\psi^{(2)}(t)\rangle$ the two particle state corresponding to standard ladder operators. The strategy to prove it is the same as before: One consider contractions between the physical ``external'' ladder operators and the operators at a single time so that the proper phases of standard evolution arise perfectly mimicking the contractions in the corresponding evaluation in standard QFT \footnote{Contractions among external operators require a regularization, either by mass uncertainty or by considering a finite time window. Importantly, this subtlety originates from the use of continuum unbounded operators and it is easily overcome. }
The same considerations hold for arbitrary physical states and normal-ordered many-particle operators acting on a single time slice.

The problem arises when  we consider expectation values of operators that are not naturally normal-ordered. In particular, such operators don't follow from a direct second quantization of sp ones but are perfectly valid operators in Fock space. The issue can be already seen in a very simple example: Consider 
$O'_F(t)=\int \frac{d^dp}{(2\pi)^d} \frac{d^dp'}{(2\pi)^d} O(\textbf{p}',\textbf{p})a(t,\textbf{p}')a^\dag(t,\textbf{p})$ such that $\langle \Omega| O'(t)_F|\Omega\rangle\neq 0$. We now obtain 
\small
\begin{equation}\label{eq:contractionOtwrong}
\begin{split}
      {}_m\langle \Psi| O'_F(t)|\Psi\rangle_m=\int &\frac{d^dp}{(2\pi)^d} \frac{d^dp'}{(2\pi)^d}  \frac{d^dk}{(2\pi)^d} \frac{d^dk'}{(2\pi)^d} \psi^\ast(\textbf{p}') O(\textbf{k}',\textbf{k})\psi(\textbf{p})\times \\&\!\!\langle \Omega|\Big[ \wick{\c1 a(E_{\textbf{p}'m},\textbf{p}') \c2 a(t,\textbf{k}') \c1 a^\dag(t,\textbf{k}) \c2 a^\dag(E_{\textbf{p}m},\textbf{p})}+ \wick{\c1 a(E_{\textbf{p}'m},\textbf{p}') \c2 a(t,\textbf{k}') \c2 a^\dag(t,\textbf{k}) \c1 a^\dag(E_{\textbf{p}m},\textbf{p})}\Big]|\Omega\rangle\,,
\end{split}
\end{equation}
\normalsize
where an internal contraction now arises in the second term. While the first term may be easily related to a standard contraction, the second term contains two new problematic elements. First, an external divergent contraction. 
Second, a contraction among local in time operators evaluated at equal times.
The problem with these terms  is not their regularization. As a matter of fact, if this computation is replaced by its discrete counterpart we still don't obtain the proper result: Schematically, the external momentum divergence corresponds to the time length $\delta(0_p)\equiv T$, while the time contraction to $\delta(0_t)\equiv 1/\epsilon$ (with $\epsilon$ the time spacing) so the overall divergence on the second term in a discrete time setting gets replaced by $\delta(0_p)\delta(0_t)\equiv T/\epsilon=N$, i.e. by the amount of time slices. One can easily and rigorously verify these considerations by developing the discrete time version of the formalism and checking the corresponding contractions (we refer the reader to the discussions in \cite{diaz2023spacetime} and Appendix \ref{app:discretequantum} for standard Fourier based discretization schemes). The important point is that the second term contains an additional factor $N$ with respect to the first. In a discrete setting this is to be expected since the first term contains contractions among Fourier-related modes yielding a factor $N^{-1/2}$ each,  while the second only contains contractions among the same modes. On the other hand, in the standard QFT formulation there is no Fourier-in-time to be considered and the two types of contractions are equivalent. Explicitly we obtain ${}_m\langle \Psi| O'_F(t)|\Psi\rangle_m=\int \frac{d^dp}{(2\pi)^d} \frac{d^dp'}{(2\pi)^d}  \psi^\ast(\textbf{p},t) O(\textbf{p}',\textbf{p})\psi(\textbf{p}',t)+\mathcal{O}(N)$ while 
$\langle \psi(t)| O|\psi(t)\rangle_m=\int \frac{d^dp}{(2\pi)^d} \frac{d^dp'}{(2\pi)^d}  \psi^\ast(\textbf{p},t) O(\textbf{p}',\textbf{p})\psi(\textbf{p}',t)+\mathcal{O}(1)$. For more complicated operators and/or states these anomalies accumulate and further misalignment follows.  In this sense, the PaW conditioning scheme of Eq.\ \eqref{eq:conditioning} does not hold for arbitrary many-body operators.

The mismatch appears precisely when internal contractions are present. This is not a minor issue: It implies that even elementary local operators, such as 
$\phi(x)\phi(x)$, cannot be mapped consistently between the two formulations. In addition, normal ordering is intrinsically vacuum dependent as the relevant notion of ``vacuum'' generally changes in the presence of interactions. Any general, consistent formalism must therefore treat internal contractions properly.

Before moving on, it is useful to comment on a few potential workarounds and caveats. A first possibility is to normal-order operators before performing the map. However, this is inevitably theory dependent, since both the vacuum and the associated normal ordering depend on the dynamics. A conceptually cleaner alternative is to abandon conditioned operators in favor of ``conditioned states'', an approach explored for many-body free theories in \cite{diaz2021spacetime,giovannetti2023geometric}. The drawback is that this process effectively washes out the spacetime structure we seek to maintain: The resulting state becomes (up to an isomorphism) a standard quantum-mechanical state. Moreover, this approach introduces other complications: The continuum-time limit produces divergences; a classical choice of reference frame must be specified (though it can be imposed a posteriori, unlike in canonical quantization), and, for interacting theories, it is not apparent how to define an appropriate physical subspace while remaining consistent to the conditioning prescription.

Our standpoint is that these obstructions reflect a more basic issue: The correct constraint structure for the many-body theory has not yet been identified. In the single particle case, the universe equation arises unambiguously as a first-class constraint already at the classical level, and its quantum version follows from Dirac quantization. By contrast, in the many-body setting the notion of  physical subspace has not been derived from an underlying constrained theory; rather, it has  been introduced as a plausible ansatz motivated by analogy and by possible generalizations of the single particle QT scheme.  As will become clearer below, once the appropriate structure is taken into account, a direct Dirac quantization is not viable. The quantum-action-based construction proposed afterwards circumvents that route and provides a solution naturally tied to the requirement that spacelike and timelike correlators be treated on the same footing, a point that is further clarified in section ~\ref{sec:ststates}. One may, of course, relax that requirement and restrict attention to standard correlations obeying Born's rule. In that case, alternative formalisms could arise, but precisely at the price of giving up part of the spacetime structure that the present construction is meant to preserve.

\section{Spacetime classical and quantum mechanics}

\subsection{Classical mechanics with spacetime Poisson Brackets}\label{sec:SCM}

In light of the above, we now introduce a classical spacetime formulation that may be regarded as the classical counterpart of the previous second quantization scheme \footnote{We remark that the formalism presented here is unrelated to applying the original Dirac scheme to a field theory. Just as we did not apply the PaW to standard QFT in section \ref{sec:secondq}. This set our approach apart from other schemes such as parameterized field theories \cite{isham1985representations,  kiefer2012quantum} and their recent PaW-like reinterpretation \cite{hoehn2023matter}. In particular, the characteristic problem of the latter arising beyond the $1+1$ dimensional case won't arise here. }. This formalism will be used to identify the missing mathematical structure in our previous proposal thus explaining the anomalies and, through additional considerations, developing the proper formalism.

We begin by introducing the classical version of the algebra in Eq.\ \eqref{eq:stfieldalg}. Namely, we 
define spacetime PBs 
\begin{equation}\label{eq:stpbs}
   \{f,g\}:=\int d^Dx\, \left(\frac{\delta f}{\delta \phi(x)}\frac{\delta g}{\delta \pi(x)}-\frac{\delta f}{\delta \pi(x)}\frac{\delta g}{\delta \phi(x)}\right)
\end{equation}
such that
\begin{equation}
    \{\phi(x),\pi(y)\}=\delta^{(D)}(x-y)\,.
\end{equation}
Notice that at this stage, nothing distinguishes space from time. We simply have a $D$-dimensional manifold, on which independent fields lie.

Let us now consider the action of a scalar field in phase-space variables
\begin{equation}
    S=\int d^Dz \left[\pi(z)\dot{\phi}(z)-\frac{\pi^2(z)}{2}-\frac{1}{2}(\nabla \phi(z))^2-\mathcal{V}(\phi(z))\right]\,.
\end{equation}
A direct calculation, employing the spacetime PBs, leads to 
\begin{align}
      \{\phi(x),S\}&=\dot{\phi}(x)-\pi(x)\\
      \{\pi(x),S\}&=\dot{\pi}(x)-\nabla^2\phi(x)+\mathcal{V}(\phi(x))\,.
\end{align}
This shows that \emph{Hamilton equations} correspond to the weak constraints
\begin{align}\label{eq:hamclasconstraints}
      \{\phi(x),S\}&=\dot{\phi}(x)-\pi(x)\approx 0\\
      \{\pi(x),S\}&=\dot{\pi}(x)-\nabla^2\phi(x)+\mathcal{V}(\phi(x))\approx 0\,.
\end{align}
Thus, in this setting, the fields corresponding to different points in spacetime are completely independent, as suggested by our definition of spacetime PBs. Yet, conventional evolution can be easily recovered by imposing that physical fields lie within the manifold specified by the simple condition 
\begin{equation}
    \{S,O\}\approx 0\,.
\end{equation}
 We recall that this is a weak equation, since we first evaluate the PBs and then we impose the condition. Clearly, a combination of Eqs.\ \eqref{eq:hamclasconstraints} leads to the Klein-Gordon equation for arbitrary potentials
\begin{equation}
    \Box \phi(x)+\mathcal{V}'(\phi(x))\approx 0\,.
\end{equation}

It is easy to show that this approach is not restricted to fields but can be applied to any classical system. In general, imposing the null PB with the action leads to Hamilton equations (see \cite{diaz2023spacetime} and the Appendix for more details).

\subsection{Quantization: the problem with Dirac's approach}\label{sec:DiracquantizationSCM}

We have shown that by defining spacetime brackets the standard Hamilton equations can be easily recovered as constraints thus allowing for an alternative formulation of classical mechanics. However, it is not difficult to show that the infinite set of constraints (one for each spacetime point) $\{\phi(x), S\}$, $\{\pi(x), S\}$ are of the second class, namely their PBs are not a linear combination of the original constraints. This leads to a very different approach than the one considered in section \ref{sec:sp}: Following Dirac, we recall that second class constraints cannot be imposed a posteriori via subspaces in Hilbert space. To illustrate this, let us recall the standard example of particle in the plane $x,y$ constrained to a line according to  $\phi_1=q_x\approx 0$, $\phi_2=p_x\approx 0$. The constraints satisfy $\{\phi_1, \phi_2\}=1\neq \alpha \phi_1+\beta \phi_2$. If we apply canonical quantization in the plane  such that $[q_x,p_x]=i$, $[q_y,p_y]=i$ (with other commutators vanishing) and then impose $q_x |\psi\rangle_{\text{phys}}=p_x|\psi\rangle_{\text{phys}}=0$ we immediately reach an incompatibility as $[q_x,p_x]|\psi\rangle_{\text{phys}}=0$ can only be satisfied if $|\psi\rangle_{\text{phys}}=0$. Notice that these conditions are also equivalent to imposing $a_x^\dag |\psi\rangle_{\text{phys}}=0$ for $[a_x, a_x^\dag]=1$ ladder operators linearly related to $q_x,p_x$.

In this scenario, where quantization leads to incompatibilities, Dirac proposes to replace PBs by Dirac brackets and only then quantize. Dirac brackets are defined as follows: Given a set of constraints $\phi_a$ so that $\{\phi_a,\phi_b\}=C_{ab}$ with $C$ an invertible matrix, the Dirac brackets between phase-space functions $f,g$ is 
\begin{equation}
\{f,g\}_{D.B}:=\{f,g\}-\sum_{a,b}\{f,\phi_a\}C^{-1}_{ab}\{\phi_b,g\}\,.
\end{equation}
In the previous example a straightforward calculation  leads to $\{q_x,p_x\}_{D.B}=0$. This means that only the variables $q_y, p_y$ are quantized and the quantum description of the particle is one-dimensional. It is crucial to realize that the number of dynamical variables has been reduced.

Let us now return to our spacetime scheme. To be more concrete, consider the free theory example of section \ref{sec:limitations}. We can  write
\begin{equation}
    S=\int \frac{d^4p}{(2\pi)^4} \, (p^0-\sqrt{\textbf{p}^2+m^2}) a^\ast(p) a(p)
\end{equation}
for
\begin{equation}\label{eq:classicalfieldsexpansion}
        \phi(x)=\int \frac{d^4p}{(2\pi)^4 \sqrt{2E_{\textbf{p}}}}\left( a(p)e^{-ipx}+a^\ast(p)e^{ipx}\right)\,,\quad
            \pi(x)=-i\int \frac{d^4p}{(2\pi)^4}\sqrt{\frac{E_\textbf{p}}{2}}\left( a(p)e^{-ipx}-a^\ast(p)e^{ipx}\right)\,,
\end{equation}
with the fields and $a(p), a^\ast(p)$ classical variables. The latter satisfy the extended bracket rule
$
    \{a(p),a^\ast(p')\}=-i(2\pi)^4\delta^{(4)}(p-p')
$
which is the classical analogue of \eqref{eq:momentumalg}. Now the constraints can be restated as
\begin{align}
   \phi_1(p):= \{a(p),S\}&=(p^0-E_{\textbf{p}})a(p)\approx 0\\
    \phi_2(p):=  \{a^\ast(p),S\}&=-(p^0-E_{\textbf{p}})a^\ast(p)\approx 0\,,
\end{align}
which must hold for all values of $p$. On the other hand, their PBs is given by  
\begin{equation}
    \{\phi_1(p),\phi_2(p')\}=i(p^0-E_{\textbf{p}})^2 (2\pi)^4\delta^{(4)}(p-p')\neq 0\,.
\end{equation}
The situation parallels our example of the particle in a plane constrained to a line.

If we ignore the problem and try to impose half of the conditions we are left, after quantization, with the constraints
\begin{equation}
 \phi_1(p)|\Psi\rangle_{\text{Phys}}=[\mathcal{S}, a(p)]|\Psi\rangle_{\text{Phys}}=0 \Rightarrow (p^0-E_{\textbf{p}})a(p)|\Psi\rangle_{\text{Phys}}=0.
\end{equation}
This equation can only hold for all $p$ if $|\Psi\rangle_{\text{Phys}}$ only contains on-shell particles. Thus we have recovered Eq.\ \eqref{eq:scommphys} from section \ref{sec:limitations}. However, a complete description of the system requires to also impose the constraints $\phi_2(p)$ which is incompatible. Let us discuss why the pathologies we described in section \ref{sec:limitations} are precisely a manifestation of having employed only half the constraints: Essentially, the constraint $\phi_1(p)$ implies $a(p)|\Psi\rangle_{\text{Phys}}\sim \delta_{p^0,E_{\textbf{p}}}a(p)|\Psi\rangle_{\text{Phys}}$ leading to $a(t,\textbf{p})|\Psi\rangle_{\text{Phys}}\sim e^{-iE_{\textbf{p}}t}a(E_{\textbf{p}},\textbf{p})|\Psi\rangle_{\text{Phys}}$, namely, up to continuum time related subtleties, $a(t,\textbf{p})$ acts on a physical state as the on-shell particle evolved to time $t$. The same idea holds for multiple annihilation operators acting on physical kets and for creation operators acting on physical bras. This effectively replaces the physical mean value of a local in time operator with the standard mean value and standard evolution.
Unfortunately, having  $a^\dag(t,\textbf{
p})|\Psi\rangle_{\text{Phys}}\simeq  e^{iE_{\textbf{p}}t}a^\dag(E_{\textbf{p}},\textbf{p})|\Psi\rangle_{\text{Phys}}$ would require $\phi_2(p)|\Psi\rangle_{\text{Phys}}=0$. This explains why only normal-ordered operators yield proper physical predictions: \emph{A naive attempt at generalizing the QT notion of physical subspace to many particles ignores a crucial part of the constraint structure that is required to recover standard QFT}.

Due to the nature of these constraints it would seem that it is not possible to properly quantize SCM. Let us convert this intuition in a precise mathematical no-go 
consequence of applying Dirac quantization scheme. 
\begin{Theorem}
     Dirac quantization scheme applied to the spacetime classical formalism leads to standard QFT. 
\end{Theorem}

\begin{proof}
In order to show this let us regularize the theory by considering a finite ($D$-dimensional) box with sides of length $T$. The Fourier modes are now defined by finite indices and frequencies $\omega_{\textbf{n}}=\frac{2\pi}{T}(n_0+n_1+...+n_d)$ with $n_i\in \mathbbm{Z}$ and $\textbf{n}=(n_0,n_1,...,n_d)$. The PBs read $\{a_{\textbf{n}},a^\ast_{\textbf{n}'}\}=-i\delta_{\textbf{n}\textbf{n}'}$. Consider the discretized constraints
\begin{equation}
\phi_{1,\textbf{n}}:=\Delta_{\textbf{n}}\,a_{\textbf{n}},\qquad 
\phi_{2,\textbf{n}}:=\Delta_{\textbf{n}}\,a^\ast_{\textbf{n}},
\end{equation}
with $\Delta_{\textbf{n}}=\frac{2\pi n_0}{T}-\sqrt{\frac{4\pi^2}{T^2}(n_1^2+n_2^2+...+n_d^2)+m^2}$. It follows that $\{\phi_{1,\textbf{n}},\phi_{1,\textbf{n}'}\}=0=\{\phi_{2,\textbf{n}},\phi_{2,\textbf{n}'}\}$, while
$
\{\phi_{1,\textbf{n}},\phi_{2,\textbf{n}'}\}=\Delta_{\textbf{n}}\Delta_{\textbf{n}'}\{a_{\textbf{n}},a^\ast_{\textbf{n}'}\}
=-i\,\Delta_{\textbf{n}}^{2}\delta_{\textbf{n}\textbf{n}'}$ and $\{\phi_{2,\textbf{n}},\phi_{1,\textbf{n}'}\}=+i\,\Delta_{\textbf{n}}^{2}\delta_{\textbf{n}\textbf{n}'}.
$
Therefore the matrix $C$ of PBs of the constraints has the block-diagonal structure
\begin{equation}
    C=\bigoplus_{\textbf{n}}\begin{pmatrix}
        0 & -i\Delta_{\textbf{n}}^2\\
        i\Delta_{\textbf{n}}^2 & 0
    \end{pmatrix}\,.
\end{equation}
Notice that on-shell there are no constraints. Instead, for $\Delta_{\textbf{n}}\neq 0$ the inverse exists and is given mode-by-mode, in components,  by
 $(C^{-1})_{(1,\textbf{n}_1)(2,\textbf{n}_2)}=\frac{i}{\Delta_{\textbf{n}_1}^2}\delta_{\textbf{n}_1\textbf{n}_2}$ and $(C^{-1})_{(2,\textbf{n}_1)(1,\textbf{n}_2)}=-\frac{i}{\Delta_{\textbf{n}_1}^2}\delta_{\textbf{n}_1\textbf{n}_2}$, with all other components vanishing. Using the definition of Dirac brackets,
$
\{f,g\}_{D.B}:=\{f,g\}-\sum_{A,B}\{f,\phi_A\}(C^{-1})_{AB}\{\phi_B,g\},
$
with $A=(\alpha,\textbf{n})$ and $\alpha=1,2$, we obtain for $\Delta_{\textbf{n}}\neq 0$ 
\begin{align}
\{a_{\textbf{n}}\,, a^\ast_{\textbf{n}'}\}_{D.B}
&=-i\delta_{\textbf{n}\textbf{n}'}
-\sum_{\textbf{n}_1,\textbf{n}_2}
\big(-i\Delta_{\textbf{n}_1}\delta_{\textbf{n}\textbf{n}_1}\big)
\left(-\frac{i}{\Delta_{\textbf{n}_1}^2}\delta_{\textbf{n}_1\textbf{n}_2}\right)
\big(-i\Delta_{\textbf{n}_2}\delta_{\textbf{n}_2\textbf{n}'}\big)\nonumber\\
&=-i\delta_{\textbf{n}\textbf{n}'}+i\delta_{\textbf{n}\textbf{n}'}=0\,.
\end{align}
We see that for all constraints for which $\Delta_{\textbf{n}}\neq 0$ Dirac scheme demotes the corresponding modes to non-dynamical. Instead, for $p^0=E_{\textbf{p}}$ (returning to continuum notation), there is actually no constraint as $\phi_1(E_\textbf{p},\textbf{p})=\phi_2(E_\textbf{p},\textbf{p})=0$. Thus the only remaining dynamical variables are on-shell modes for which the DBs and PBs are equal.
\end{proof}

Let us make the statement of the theorem more explicit: Since the only dynamical variables are those $a(p),a^\ast(p)$ for which $p^0=E_{\textbf{p}}$ we can write for any integral involving the DBs can be simplified as $\int \frac{d^Dp}{(2\pi)^D} \frac{d^Dp'}{(2\pi)^D} f(p,p')\{a(p),a^\ast(p')\}_{D.B.}
=-i\int \frac{d^dp}{(2\pi)^d} f(E_{\textbf{p}},\textbf{p},E_{\textbf{p}},\textbf{p})$. In particular,
by using Eq.\ \eqref{eq:classicalfieldsexpansion}, which expands the fields as integrals of modes, we obtain 
\begin{equation}
    \{\phi(t,\textbf{x}),\pi(t, \textbf{y})\}_{D.B.}=\delta^{(d)}(\textbf{x}-\textbf{y)}\,,
\end{equation}
which quantized leads to the standard \emph{equal-time} commutation relations and with $\{\phi(t,\textbf{x}),\phi(t', \textbf{y})\}_{D.B}\neq 0$ for $t\neq t'$. Thus fields at different points in time become dependent, and connected by the usual rules of evolution. While the Dirac approach is perfectly consistent, leading us directly to standard QFT, it  removes completely the geometrical spacetime picture that we proposed classically. Notice also that the particular choice of foliation made by Dirac brackets is inherited by the choice made by the Legendre transform that defines the action \footnote{In \cite{diaz2023spacetime} the authors describe how SQM also allows one to keep that choice arbitrary by defining a quantum foliation, which might be interpreted as a quantum reference frame for QFTs. See also section \ref{sec:concl}}. Notably, if one quantizes according to Dirac scheme, a particular observer must be chosen before quantization and all the spacetime symmetries are hidden, just as in the usual canonical quantization of field theories.

\subsection{Quantization: The quantum action approach}\label{sec:quantumactionquantization}

The previous no-go theorem suggests that QM and relativity are fundamentally incompatible: If one can only quantize  ``genuine'' dynamical variables, and those variables have to be determined according to a particular choice of what is meant by ``time'', there is no hope to develop a quantum geometric picture of fields in spacetime where covariance is explicit. Even if we develop a classical spacetime picture first, as in section \ref{sec:SCM},  the nature of the constraints makes it impossible to work with conventional (extended) quantum states as the one proposed in section \ref{sec:limitations}. The incompatibilities we discussed in section \ref{sec:limitations} are indeed a symptom of this deeper problem.

Notably there is a completely different way of thinking about the issue that circumvents it entirely and which brings QM and relativity conceptually closer. Possibly, and perhaps surprisingly, the easiest path to solve the problem is to borrow inspiration from \emph{quantum computing}: Since standard QM operates linearly, the only general way to obtain quantities that are functions of powers of the entries of a quantum state is to consider the tensor product of copies of the state. 
The most basic example is the SWAP test that allows one to compute the purity of a state as follows ${\rm tr}[\rho^2]={\rm Tr}[\rho \otimes \rho SWAP]$ (since SWAP is unitary this mean value can be easily estimated by means of a DQC1 protocol \cite{knill1998power}). This relation is straightforwardly verified as follows
\begin{align*}
    {\rm Tr}[\rho \otimes \rho \, \text{SWAP}]&=\sum_{i_0,i_1}\langle i_0 i_1|\rho \otimes \rho \,\text{SWAP} |i_0i_1\rangle=\sum_{i_0,i_1}\langle i_0 i_1|\rho \otimes \rho  |i_1i_0\rangle\\&=\sum_{i_0,i_1}\langle i_1|\rho|i_0\rangle \langle i_0|\rho |i_1\rangle=\sum_i \langle i|\rho^2|i\rangle={\rm tr}[\rho^2]\,,
\end{align*}
where we used a complete basis $|i\rangle$ of the original Hilbert space.

How is this comment related to our discussion about SQM?
The essential point is that for bosonic fields imposing a spacetime algebra is equivalent to define fields acting on the tensor product \emph{in time} of  the traditional Hilbert space of QFT. We discussed this point in section \ref{sec:secondq} by discretizing the spacetime algebra \ref{eq:basicalg}. If two copies are considered, we can use the SWAP operator to relate observables that live in $\mathcal{H}=h\otimes h$ to observables in $h$, for $h$ the conventional Hilbert space of QFTs. For $N$ slices, the generalization is provided by the operator $e^{i\epsilon \mathcal{P}}$ satisfying
\begin{equation}
    e^{i\epsilon \mathcal{P}}|i_0i_1\dots i_{N-1}\rangle=|i_{N-1}i_0i_1\dots\rangle\,,
\end{equation}
namely $e^{i\epsilon \mathcal{P}}$ implements single step time-translations \emph{across} slices. The operator is unitary and can be written as a product of different SWAP operators. Here $\mathcal{P}$ is the generator of these translations and the factor $\epsilon$ indicates that we are translating a single time step such that e.g. $(e^{i\epsilon \mathcal{P}})^2=e^{i2\epsilon \mathcal{P}}$ translates two times and so on. It is a straightforward exercise to show that ${\rm Tr}[e^{i\epsilon \mathcal{P}}\otimes_t O^{(t)}_t]={\rm tr}[\hat{T}\prod_t O^{(t)}]$. On the left hand side we are considering tensor product of operators while on the right hand side we have a standard composition of time ordered operators. Thus the time translation operator allows us to map between quantities in $\mathcal{H}=h^{\otimes N}$ and quantities in $h$ \cite{diaz2025spacetime}.

\def\arraystretch{1.5}
 \begin{table}
     \centering
     \begin{tabular}{|c|}
     \hline
         \textbf{Quantum-action-based quantization}   \\\hline\hline
         Promote the complete spacetime PBs to spacetime commutators: $\{,\}\to \frac{1}{i\hbar} [,]$\\ \hline
         Promote all functions in the spacetime phase space to operators. This includes the action $S\to \hat{\mathcal{S}}$      \\\hline
         Quantum mechanical quantities are defined by $\langle \dots \rangle:={\rm Tr}[e^{i\hat{\mathcal{S}}/\hbar}\dots]$   \\ \hline
         The classical constraints are satisfied within $\langle \dots \rangle$ and yield standard unitary evolution\\\hline
          Equal time correlators of hermitian operators are equal to standard quantum expectation values\\\hline
     \end{tabular} 
     \caption{Steps of the quantum action quantization scheme applied to the spacetime classical mechanics formalism and basic properties.}
     \label{tab:quantumaction}
 \end{table}

Notably, the operator $\mathcal{P}$ for continuum time theories is an object that we already defined: 
\begin{equation}
    \mathcal{P}=\int d^Dx\, \pi(x) \dot{\phi}(x)\,,
\end{equation}
which is the first part of the quantum action operator we obtained via second quantization of the operator $p^0$, the generator of time translations in QT schemes. One can readily verify that 
\begin{equation}
    e^{i\tau \mathcal{P}}\phi(t,\textbf{x}) e^{-i\tau \mathcal{P}}=\phi(t+\tau,\textbf{x})\,,\;\;\;\;  e^{i\tau \mathcal{P}}\pi(t,\textbf{x}) e^{-i\tau \mathcal{P}}=\pi(t+\tau,\textbf{x})\,,
\end{equation}
showing that indeed $\mathcal{P}$ implements time translations \emph{across} slices (we recall that $t$ here does not correspond to unitary evolution but is just a label). 
Considering that for discrete time $e^{i\epsilon \mathcal{P}}$ can be used to connect mean values between copies of a Hilbert space with mean values in a single copy it is natural to consider quantities of the form ${\rm Tr}[e^{i\tau \mathcal{P}}\dots]$, where in principle we take the limit $\tau\to 0$ to mimic the discrete time case. If we in addition consider the complete quantum action, $e^{i\tau \mathcal{S}}=e^{i\tau \mathcal{P}}e^{-i\tau \int dt\, H}$ we see that the first piece generates translations across time, while the second generates time translations within slices (standard evolution) leading us to a mathematically simple but conceptually important theorem.

\begin{Theorem}\label{th:constraintseis}
Consider the ``mean value'' $\langle \dots \rangle_\tau :=\frac{{\rm Tr}[e^{i\tau \mathcal{S}}\dots ]}{{\rm Tr}[e^{i\tau \mathcal{S}}]}$.
    The constraints $[\mathcal{S},O]\approx 0$ are satisfied within $\langle \dots \rangle_\tau$.
\end{Theorem}
\begin{proof}
The proof is direct:
    \begin{equation}
    {\rm Tr}[e^{i\tau\mathcal{S}}[\mathcal{S},O]]={\rm Tr}[e^{i\tau\mathcal{S}}\mathcal{S}O]-{\rm Tr}[e^{i\tau\mathcal{S}}O\mathcal{S}]={\rm Tr}[e^{i\tau\mathcal{S}}\mathcal{S}O]-{\rm Tr}[Se^{i\tau\mathcal{S}}O]={\rm Tr}[e^{i\tau\mathcal{S}}\mathcal{S}O]-{\rm Tr}[e^{i\tau\mathcal{S}}\mathcal{S}O]=0
\end{equation}
where we used the cyclicity of the trace and that $[e^{i\tau\mathcal{S}},\mathcal{S}]=0$.
\end{proof}
 Then, as long as we consider only quantities within the trace, the constraints are automatically satisfied. In this sense, the ``quantum action approach'' we are introducing circumvents the problem with second class constraints entirely. One can verify that while the classical constraints correspond to Hamilton equations, their quantum counterpart is instead related to Heisenberg equations (this is more easily proven for discrete time; see Appendix \ref{app:discretequantum}). 
 Moreover, since this holds for any value of $\hbar$, and we might evaluate the trace in the eigenbasis of fields, we can make a path integral like argument to recover the classical PB constraints from the $\hbar$ limit. This shows a consistency between this quantization scheme and SCM. These arguments are described in more detail in the Appendix \ref{app:discretequantum} where we also consider the discrete time version of constraints and initial/final conditions. The general quantization picture we are introducing is summarized in Table \ref{tab:quantumaction}.

Let us now show that quantities between brackets $\langle \dots \rangle_\tau$ are indeed related to standard QFT observables. Let us first recall that quadratic operators are fully determined by their basic contractions (Wick's theorem). For a diagonal quadratic operator, the only non-trivial correlator is
\begin{equation}\label{eq:contr}
    \langle a^\dag_k a_l\rangle:=\frac{{\rm Tr}\big[\exp(-\sum_i \lambda_i a^\dag_i a_i)a^\dag_k a_l\big]}{{\rm Tr}\big[\exp(-\sum_i \lambda_i a^\dag_i a_i)\big]}=\frac{1}{\exp(\lambda_k)-1}\delta_{kl}\,.
\end{equation}
Here the indices $k,l$ are ``space-like separated'', in the sense that the operators $a^\dag_k$, $a_l$ are not evolved in the given reference frame and correspond to orthogonal modes. Let us now apply this basic result to the free Klein-Gordon action of Eq.\ \eqref{eq:freekgaction}:
    \begin{align}\label{eq:momcorr}
  \langle a^\dag(p)a(k)\rangle_\tau&=\frac{{\rm Tr}\big[\exp\big\{i\tau \int \frac{d^Dq}{(2\pi)^D}\, (q^0-E_\textbf{q}+i\epsilon)a^\dag(q)a(q)\}a^\dag(p)a(k)\big]}{{\rm Tr}\big[\exp \big\{i\tau \int \frac{d^Dq}{(2\pi)^D}\, (q^0-E_\textbf{q}+i\epsilon)a^\dag(q)a(q)\big \}\big]}
  \\&=\frac{1}{\exp \{-i\tau (p^0-E_\textbf{p}+i\epsilon)\}-1}(2\pi)^D\delta^{(D)}(p-k)\,,
\end{align}
where added a small imaginary part to the part of the action corresponding to standard time evolution (in section \ref{sec:ststates} and in the last part of Appendix \ref{app:discretequantum} we clarify how this is related to the action effectively projecting onto the vacuum state).  We can use this basic expression to show by direct evaluation that \footnote{A discussion on the $\tau$ regularization we are using can be found in \cite{diaz2021path}. Therein  it is also shown that $\tau$ can be related to the quotient between the Hilbert space and Feynman measures. We also recall that the small $\tau$ limit is replacing the role of $\epsilon$ in discrete time settings. For discrete time the previous relations hold for finite $\epsilon$ and insertion of operators at times commensurable with $\epsilon$, see Appendix \ref{app:discretequantum}. }
\begin{equation}\label{eq:aacorrel}
 \lim_{\tau\to 0}  \langle \sqrt{\tau} a(t,\textbf{p}) \sqrt{\tau} a^\dag(E_{\textbf{k}},\textbf{k})\rangle_\tau=e^{-iE_{\textbf{k}}t}(2\pi)^d \delta^{(d)}(\textbf{p}-\textbf{k})=\langle 0|a_H(t,\textbf{p})a^\dag(\textbf{k})|0\rangle\,,
\end{equation}
and 
\begin{equation}
 \lim_{\tau\to 0}  \langle  \sqrt{\tau} a(E_{\textbf{k}},\textbf{k}) \sqrt{\tau} a^\dag(t,\textbf{p})\rangle_\tau=e^{iE_{\textbf{k}}t}(2\pi)^d \delta^{(d)}(\textbf{p}-\textbf{k})=\langle 0|a(\textbf{k})a^\dag_H(t,\textbf{p})|0\rangle\,.
\end{equation}
we see that now both $a(t,\textbf{p})$ and $a^\dag(t,\textbf{p})$ behave like proper physical operators within brackets, with $|0\rangle$ the usual ground state of the free Klein-Gordon Hamiltonian $H$ and $a_H(t,\textbf{p}):=e^{iHt}a(\textbf{p})e^{-iHt}$ a conventional ladder operator in the Heisenberg picture. Moreover, the previous problematic internal contractions get replaced by 
\begin{equation}
 \lim_{\tau\to 0}  \langle \sqrt{\tau} a(t,\textbf{p}) \sqrt{\tau} a^\dag(t',\textbf{k})\rangle_\tau=e^{-iE_\textbf{p}(t-t')}\theta(t-t')(2\pi)^d \delta^{(d)}(\textbf{p}-\textbf{k})=\langle 0|\hat{T}a_H(t,\textbf{p})a^\dag_H(t',\textbf{p})|0\rangle\,,
\end{equation}
where the heaviside step function follows from $e^{-iE_\textbf{k}(t-t')}\theta(t-t')=\int \frac{dk^0}{k^0-E_{\textbf{k}}+i\epsilon}e^{-ik^0(t-t')}$.

As a consequence of the previous discussion we can safely conclude that we have identified the right structure to generalize the PaW-like conditioning to many particles. However, this generalization is indeed non-trivial: We are  no longer making use of physical states in the traditional sense. Instead, we are considering mean values with respect to the exponential of the action. In other words, it would seem that by abandoning the Dirac quantization scheme we are also abandoning the standard notion of quantum state.

\subsection{Interacting field theories within the quantum action approach}\label{sec:interactions}

Before discussing the implications of our scheme regarding quantum states, let us briefly discuss how the formalism provides a very natural way to introduce standard observables in QFT, arising both for free theories and in perturbative approaches to interacting theories.

While in the previous section we focused on ladder operators, one can easily obtain an expression for the mean value of different field insertions: \begin{align}\label{eq:feynproptau}
   \langle\phi(x)\phi(y)\rangle_\tau= \frac{1}{\tau}\int \frac{d^Dp}{(2\pi)^D}\frac{i}{p^2-m^2+i\epsilon}e^{-ip (x-y)}+\mathcal{O}(\tau)
\end{align}
which follows from Eq.\ \eqref{eq:momcorr} and $\frac{i}{p^0-E_\textbf{p}+i\epsilon}-\frac{i}{p^0+E_\textbf{p}-i\epsilon}\equiv 2E_\textbf{p}\frac{i}{{p^2-m^2+i\epsilon}}$. One immediately recognizes the momentum representation of Feynman propagator \cite{peskin2018introduction} which allows us to write
\begin{equation}\label{eq:feynprop}
 \lim_{\tau\to 0}  \langle\sqrt{\tau}\phi(x)\sqrt{\tau}\phi(y)\rangle_\tau= \langle 0| \hat{T} \phi_H(x)\phi_H(y)|0\rangle = \;\; \includegraphics[height=5.5pt]{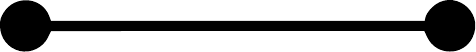}
\,.
\end{equation}
On the right-hand side, $\phi_H(\textbf{x},t):=e^{iHt}\phi(\textbf{x})e^{-iHt}$ is a conventional field operator in the Heisenberg picture and $|0\rangle$ is the usual ground state of the free Klein Gordon Hamiltonian $H$.  On the left-hand side, the operators are not evolved, instead, their ``position in time'' determines the amount of evolution that is recovered upon evaluation. We have included a diagrammatic representation of the propagator as a solid line (since we are considering a real scalar field it is not necessary to include arrows).

Since higher momenta can be computed via Wick's theorem we have a clear prescription to fully recover all time-ordered correlation functions: We simply need to consider the expectation value with respect to the exponential of the action of operators inserted at different spacetime points. Notably, we have recovered all of the main ingredients employed in standard QFT to extract physical predictions, such as scattering amplitudes and decay rates \cite{peskin2018introduction}.

As a matter of fact, the previous considerations also provide us  the proper treatment of interacting theories: We just needs to consider the previous quantum action with the addition of interacting terms, and fields that are rescaled properly. For example, a $\lambda\phi^4$ theory is given by
\begin{equation}
    \mathcal{S}_\tau=\tau\int d^Dz \left[\pi(z)\dot{\phi}(z)-\frac{\pi^2(z)}{2}-\frac{1}{2}(\nabla \phi(z))^2-\tau\frac{\lambda}{4!}\phi^4(\tau)\right]\,.
\end{equation}
Then, the \emph{interacting} correlation functions are obtained from $\langle \dots \rangle^{(\text{int})}_{\tau}\equiv \frac{{\rm Tr}[e^{i\mathcal{S}_\tau}\dots]}{{\rm Tr}[e^{i\mathcal{S}_\tau}]}$ with the index ``int'' indicating that we are considering the full interacting action. Since we are only interested in the small $\tau$ limit we can separate the free and interacting parts (Trotter approximation), and then expand the latter perturbatively to write
\begin{equation}\label{eq:perturb}
    \langle \dots \rangle^{(\text{int})}_{\tau}\simeq \langle \left(1-i\tau^2 \frac{\lambda}{4!}\int d^Dz\, \phi^4(z)-\frac{(\tau^2 \lambda)^2}{4!^2} \int d^Dz d^Dw\, \phi^4(z) \phi^4(w)+\dots \right) \rangle_{\tau}
\end{equation}
where only connected diagrams are to be considered (the denominator should be expanded as well and the usual ``vacuum bubble'' cancellation follows). Within $\langle \dots \rangle_\tau$ we may use the basic two-point contractions to compute any mean value, as it follows from Wick's theorem applied to the gaussian operator $e^{i\tau\mathcal{S}_{\text{free}}}$ \footnote{We remark that in the present formalism one does not need to derive a Wick's theorem of time-ordered contractions. We only make use of the basic ``algebraic'' Wick's theorem applied to free/gaussian quantum actions. The time-dependent Wick's theorem emerges as a consequence of the properties of the action.}. To give a few basic examples, when evaluating $\langle \phi(x_1)\phi(x_2)\phi(x_3)\phi(x_4)\rangle^{(\text{int})}_{\tau}$ the first order term in Eq.\ \eqref{eq:perturb}
leads to contractions of the form
\begin{equation}
\langle  \wick{\c1\phi(x_1)\c2\phi(x_2)\c3\phi(x_3)\c4\phi(x_4)\c1\phi(z)\c2\phi(z)\c3\phi(z)\c4\phi(z)}\rangle_\tau\;=\; \raisebox{-0.4\height}{\includegraphics[height=35pt]{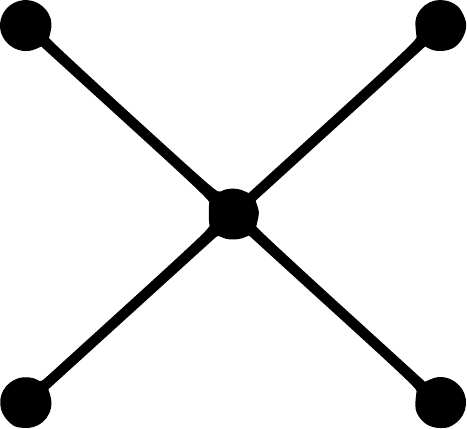}}
\,,
\end{equation}
while the second order term leads e.g. to 
\begin{equation}
\langle  \wick{\c1\phi(x_1)\c2\phi(x_2)\c3\phi(x_3)\c4\phi(x_4)\c1\phi(z)\c3\phi(z)\c5\phi(z)\c6\phi(z)\c5\phi(w)\c6\phi(w)\c2\phi(w)\c4\phi(w)}\rangle_\tau \;=\; \raisebox{-0.4\height}{\includegraphics[height=35pt]{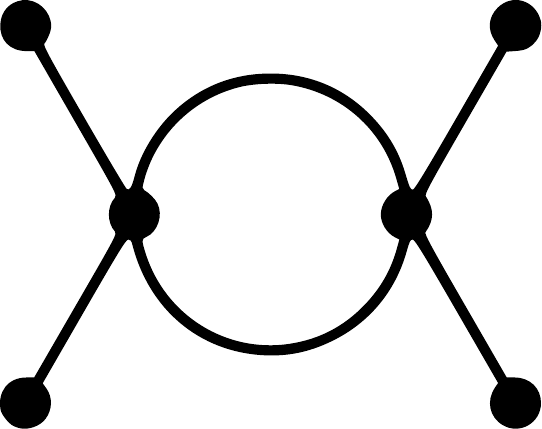}}
\end{equation}
\noindent and to other permutations on external legs of this one-loop diagram (corresponding to the three different scattering channels). One can verify that for each propagator, a factor $\tau$ is implicitly involved in these contractions so that the proper small $\tau$ limit gives the corresponding standard QFT result with the Feynman diagrams representing the same quantities. This also indicates that for small $\tau$ the divergences appearing in loop diagrams in the current formalism match precisely the divergences in the standard formulation. In addition, the formalism allows one to regulate interacting theories from the onset. This in particular suggest novel ways to explicitly control the Lorentz symmetry breaking induced by regularization, with the symmetry of space and time preserved in the process (see details in Appendix \ref{app:continuumtimelimit}).

In summary, in order to obtain correlation functions of interacting theories one just need to consider expectation values of field insertions with respect to interacting quantum actions. We recall that correlation functions are one of the main ingredients of QFTs as they, in particular, allow one to compute S-matrix elements by means of the Lehmann–Symanzik–Zimmermann (LSZ) reduction formula  \cite{lehmann1955formulierung}. 
Similar considerations hold for fermions and boson-fermion interactions (see Appendix \ref{app:fermions} and the recent results in \cite{diaz2025spacetime}).

\subsection{Quantum actions and the concept of quantum state in spacetime}\label{sec:ststates}

Let us now discuss  one of the main consequences of abandoning Dirac quantization and replacing it with our proposal. As we have seen, we cannot impose constraints on standard states to define a physical subspace. Instead, we introduced ``expectation values'' of local operators computed with respect 
to the exponential of the action. Notably, this somewhat unusual proposal can be related to the ongoing discussion on the possibility of defining quantum states over time \cite{horsman2017can, fitzsimons2015quantum, cotler2018superdensity, fullwood2024operator, diaz2025spacetime, milekhin2025observable, guo2025spacetime}, a discussion that so far has been mostly orthogonal to the QT scenario. In fact, while QT schemes deal with a quantum time operator, the states over time proposals deal with a tensor product structure across time. 
As we have shown in section \ref{sec:secondq} these two notions are linked via second quantization of QT schemes.

In order to briefly compare with the literature, let us introduce the subject for two time-slices ($N=2$) and arbitrary quantum systems so that $\mathcal{H}=h\otimes h\equiv \otimes_{t=0}^1h$ for $h$ the standard Hilbert space of the system under consideration. As we discussed in section \ref{sec:quantumactionquantization} for $N=2$ we have $e^{i\mathcal{S}}= \text{SWAP}\,U\otimes U$ with $U=e^{-iH t}$ (here we take $t=T/N$). If we also introduce an initial state we can define $\mathcal{R}:=(|\psi\rangle \langle \psi|(U^\dag)^2 \otimes \mathbbm{1}) \text{SWAP}\,U\otimes U=(|\psi\rangle\langle \psi| \otimes \mathbbm{1})\,\text{SWAP}\,U\otimes U^\dag$ where $(U^\dag)^2=e^{i H T}$. With these definitions it is easy to show that
\begin{equation}\label{eq:partialtr}
    {\rm Tr}_1[{\mathcal{R}}]=|\psi\rangle \langle \psi|\,, \quad     {\rm Tr}_0[{\mathcal{R}}]=|\psi(t)\rangle \langle \psi(t)|\,.
\end{equation}
We see that the partial trace of $\mathcal{R}$ over a time slice yields the quantum state evolved to the complement time. 
These are precisely the requirements proposed in \cite{horsman2017can} for a quantum state over time, here recovered from the exponential of the quantum action multiplied on the left by the initial state. In \cite{horsman2017can} the hermiticity condition was also imposed. As a matter of fact, among other proposals, the authors considered as a candidate for state over time the operator
$(\mathcal{R}+\mathcal{R}^\dag)/2$. It turns out that relaxing the hermiticity condition is quite insightful. To show this let us first notice that Eq.\ \eqref{eq:partialtr} does not fully determine $\mathcal{R}$. Instead, one can uniquely define $\mathcal{R}$ by imposing all spacetime correlators ${\rm Tr}[\mathcal{R}A\otimes B]$. With our definition we have ${\rm Tr}[\mathcal{R}A\otimes B]=\langle \psi|B(t)A|\psi\rangle$ while ${\rm Tr}[\mathcal{R}^\dag A\otimes B]=\langle \psi|AB(t)|\psi\rangle$. We see that $\mathcal{R}$ corresponds to time-ordered evolution, $\mathcal{R}^\dag$ corresponds to anti-time ordered ones. This is particularly useful as it implies ${\rm Tr}[(\mathcal{R}-\mathcal{R}^\dag)A\otimes B]=\langle \psi|[B(t),A]|\psi\rangle$ showing that \emph{non-hermiticity is a signature of causality}. This is in striking contrast with standard quantum states defined at a given time but separated in space, for which necessarily $[A,B]=0$ (for all local operators on different subsystems) implying $\rho-\rho^\dag=0$. 
Moreover, if we now define $\mathcal{R}$ for arbitrary many time and space slices (see  \cite{diaz2025spacetime} for details):  \begin{equation}\label{eq:Rdiscrete}
    \mathcal{R}:=(|\psi\rangle \langle \psi|e^{-iHT}\otimes_{t=1}^{N-1} \mathbbm{1})\,e^{i\epsilon \mathcal{P}}\otimes_{t=0}^{N-1} e^{-i\epsilon H}\equiv |\psi\rangle_0 \langle \psi, -T| \;e^{i\mathcal{S}}\,,
\end{equation} with $|\psi\rangle$ and $H$ acting on $\otimes_x h_x$,  ${\rm Tr}_{\text{all-spacetime}}[(\mathcal{R}-\mathcal{R}^\dag) A_x B_y]={\rm Tr}_{xy}[(\mathcal{R}_{xy}-\mathcal{R}^\dag_{xy}) A_x B_y]$  where the sub-indices $x,y$ in the operators indicate the Hilbert space in which they act while $\mathcal{R}_{xy}:={\rm Tr}_{\overline{xy}}[\mathcal{R}]$ where $\overline{xy}$ denotes the complement of $x \cup y$. Whether $\mathcal{R}_{xy}$ is a standard quantum state (or how close is to one) will be determined by the causal relation among $x,y$ for a given theory. In particular, equal-time partial traces lead to \emph{standard quantum states for all theories}, at least for those defined by standard QM where the evolution sector of the action is given by a tensor product of evolutions (as in \eqref{eq:Rdiscrete}). Below we discuss how these comments relate to the concept of \emph{microcausality} in QFT; see also figure \ref{fig:lightcone}.
Because of the previous considerations it is natural to think of $\mathcal{R}$ as a proper generalization of the notion of state to the spacetime setting. 
Notice also that $\mathcal{R}$ is the \emph{unique} object that contains both timelike and spacelike correlators as expectation values of local operators in $\mathcal{H}$. At the same time, $\mathcal{R}$ is directly related to the exponential of the quantum action suggesting a direct interpretation: The SQM formalism provides a Hilbert space embedding of the ``sum over histories'' of Feynman (see also \cite{diaz2021path}).

Let us also mention some interesting entropic properties of spacetime states that one can easily introduce for discrete time: One can easily show that $\mathcal{R}^N=\otimes_t |\psi(t)\rangle \langle \psi(t)|$ which is a projector for pure states. This leads to $\mathcal{R}$ having a single eigenvalue $1$ and other eigenvalues equal to zero, implying  ${\rm Tr}[\mathcal{R}^k]=1$ for all $k$, so that all pseudo-entropies \cite{harper2025non} of $\mathcal{R}$ vanish as long as the initial state is pure and the evolution is unitary. While it might seem unusual to define entropies for non-hermitian operators, the concept of pseudo-entropy has indeed been related to timelike entanglement in the context of the dS-CFT correspondence \cite{harper2023timelike}, and in particular, for $N=2$ it was recently shown \cite{milekhin2025observable} that a version of $\mathcal{R}$ (see \cite{diaz2025spacetime} for a direct comparison) provides a microscopic definition of such timelike entanglement in CFTs, later expanded in \cite{guo2025spacetime}.  For a recent discussion about non-hermitian generalizations of QM we refer to \cite{harper2025non}.
Moreover, for thermal correlators it was also shown that spacetime states $\mathcal{R}$ satisfy a variational principle generalizing the standard classical principle of minimal action by the addition of a Von-Neumann pseudo-entropy \cite{diaz2025spacetime}.

\begin{figure}[t!]
    \centering
    \includegraphics[width=0.8\linewidth]{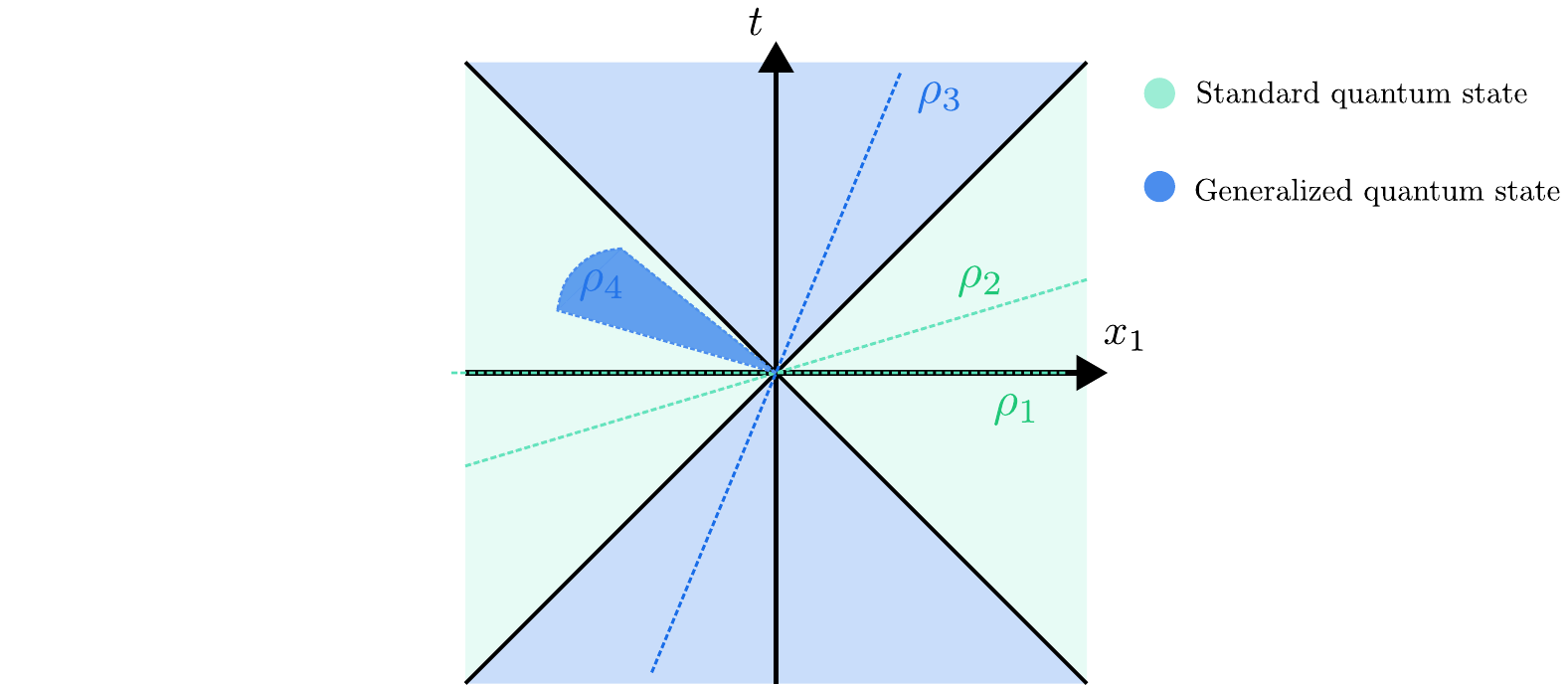}
    \caption{Different causality induced behavior of partial traces of a spacetime state $\mathcal{R}$ of a relativistic theory. For the observer  of the Minkowski diagram canonical quantization is imposed at fixed time and leads to a quantum state $\rho_1$. Another observer quantizing on a different ``tilted'' surface also defines its corresponding notion of state $\rho_2$. While these two states are conventionally associated with different (isomorphic) Hilbert spaces, from the spacetime perspective both $\rho_1$ and $\rho_2$ can be recovered from partial traces of a single object $\mathcal{R}$. 
    One can also consider an operator $\rho_3$ defined on a timelike surface also determined by $\mathcal{R}$. This operation does not yield a standard quantum state as causality requires $\rho_3^\dag \neq \rho_3$. Even more generally, if one considers a region which is not entirely spacelike (for some pair of points the interval is timelike) the corresponding object $\rho_4$ is also not a standard quantum state. }
    \label{fig:lightcone}
\end{figure}

Let us now return to the continuum formalism we developed throughout the manuscript and reinterpret it in light of the spacetime state discussion.  In order to understand why in the continuum case we don't need to introduce an initial state explicitly we notice first that ${\rm Tr}[(e^{i\epsilon \mathcal{P}}\otimes_t e^{-i\epsilon H})\otimes_t O^{(t)}_t]={\rm tr}[e^{iHT}\hat{T} \prod_t O_H^{(t)}(t)]$. The key point is that on the right hand side an operator $e^{iHT}$ appears, showing that the expectation value of the exponential of the action corresponds to propagators (with e.g. ${\rm Tr}[|\psi\rangle_0\langle \varphi| (e^{i\epsilon \mathcal{P}}\otimes_t e^{-i\epsilon H})]=\langle \varphi| e^{iHT}|\psi\rangle$). In the limit of large $T$ with a small imaginary part, the operator $e^{iHT}$ becomes a projector onto the vacuum, a standard assumption that leads to Feynman prescription \cite{weinberg2005quantum,peskin2018introduction} and which let us omit the term  $|\psi\rangle_0\langle \varphi|$ altogether if we are interested in vacuum states. This means that 
\begin{equation}
    \mathcal{R}=\frac{e^{i\mathcal{S}_\tau}}{{\rm Tr}[e^{i\mathcal{S}_\tau}]}\,,
\end{equation}
is already \emph{the spacetime state corresponding to the vacuum}, in agreement with ${\rm Tr}[\mathcal{R}O(t)]\equiv \langle 0| O_H(t)|0\rangle$ (here ``$\equiv$'' denotes the proper $\tau$ limit explained in the previous sections) and provided we follow Feynman's prescription as in Eqs.\ (\ref{eq:momcorr}, \ref{eq:feynproptau}). 
Let us now notice that $\mathcal{R}^\dag$ leads to an anti time order prescription as ${\rm Tr}[\mathcal{R}^\dag \phi(x) \phi(y)]={\rm Tr}[\mathcal{R}\, \phi(x) \phi(y)]^\ast$ while $\langle 0|\phi_H(x)\phi_H(y)|0\rangle^\ast=\langle 0|\phi_H(y) \phi_H(x)|0\rangle$ (this also follows by noting that the dagger interchanges $i\epsilon \to -i\epsilon$). As a consequence,  \begin{equation}
    {\rm Tr}[(\mathcal{R}-\mathcal{R}^\dag)\phi(x)\phi(y)]\equiv\langle 0| [\phi_H(x),\phi_H(y)] |0\rangle= [\phi_H(x),\phi_H(y)]\,,
\end{equation} where we assumed $x^0\geq y^0$ and that the commutator is a $c$-number (as in free theories). 
In standard QFT this \emph{unequal} time commutator determines if the field at $x$ and the field at $y$ are causally connected, and, in particular, for relativistic QFTs one finds a vanishing commutator iff $|x-y|$ is spacelike. This is the standard statement of microcausality, and provides the basis to prove that QFT, in its standard non-manifestly covariant form, is compatible with relativity.  In the spacetime approach this quantity appears under a new but related light:  
The causal connection between fields at two spacetime points determine whether the corresponding $\mathcal{R}$, reduced to a region containing both $x$ and $y$, can be a standard quantum state, as depicted in figure \ref{fig:lightcone}. Since $\mathcal{R}$ encompasses all spacetime, it is clear that this basic object from which standard quantum states may be derived cannot be itself a standard quantum state.

The previous $\mathcal{R}$ is the spacetime version of the vacuum state. In order to define spacetime states corresponding to excited states we need to consider instead
$\mathcal{R}=\frac{a(T,\textbf{k}_1)a(T,\textbf{k}_2)\dots e^{i\mathcal{S}_\tau} a^\dag(-T,\textbf{p}_1)a^\dag(-T,\textbf{p}_2)\dots}{{\rm Tr}[a(T,\textbf{k}_1)a(T,\textbf{k}_2)\dots e^{i\mathcal{S}_\tau} a^\dag(-T,\textbf{p}_1)a^\dag(-T,\textbf{p}_2)\dots]}$, namely we insert annihilation operators on the l.h.s. and creation operators on the r.h.s. of the exponential of the action for large $T$. This mimics the insertion $|\psi\rangle_0\langle \varphi|$ we discussed in the discrete case. As a matter of fact, ${\rm Tr}[\mathcal{R}O(t)]\equiv  \langle \varphi| O_H(t)|\psi\rangle$ for 
$|\psi\rangle=a^\dag(\textbf{p}_1)a^\dag(\textbf{p}_2)\dots|0\rangle$, $|\varphi\rangle=a^\dag(\textbf{k}_1)a^\dag(\textbf{k}_2)\dots|0\rangle$  (we are assuming as usual that for large $T$ the evolution of ladder operators is essentially free). Notice that the local operator $O(t)$ here is arbitrary, and in particular no normal order is needed. This confirms that we have circumvented the problems identified in section \ref{sec:limitations} which arise from a naive second quantization approach based on standard quantum states.

Moreover, when all the momenta $\textbf{p}$ and $\textbf{k}$ differ, we can make use of standard limits to replace the last expression of $\mathcal{R}$ with 
\begin{equation}
\mathcal{R}'\propto a(E_{\textbf{k}_1},\textbf{k}_1)a(E_{\textbf{k}_2},\textbf{k}_2)\dots a(E_{\textbf{k}_{n'}},\textbf{k}_{n'}) e^{i\mathcal{S}_\tau} a^\dag(E_{\textbf{p}_1},\textbf{p}_1)a^\dag(E_{\textbf{p}_2},\textbf{p}_2)\dots a^\dag(E_{\textbf{p}_n},\textbf{p}_n)\,.
\end{equation} Now all the ladder operators are  \emph{on-shell}. Here we are omitting an overall constant discussed in the Appendix \ref{app:Smatrix}. This unnormalized $\mathcal{R}'$ leads to
\begin{equation}
    {\rm Tr}[\mathcal{R}']= \langle \textbf{k}_1 \textbf{k}_2\dots \textbf{k}_{n'}|S|\textbf{p}_1\textbf{p}_2\dots \textbf{p}_{n}\rangle\,,
\end{equation}
 i.e. we have found an expression for the $S$-matrix elements. Hence, $\mathcal{R}'$ is the \emph{spacetime transition-state corresponding to non-trivial scattering amplitudes}. We provide a proof of this expression in Appendix \ref{app:Smatrix} where we also relate the proportionality constant to the LSZ reduction formula and explicitly compute a simple scattering amplitude by employing the formalism. Notice that while the full idea of physical state inspired by QT does not hold, the notion of physical particle  associated with on-shell ladder operator does indeed defines the physical particles in scattering processes.\\

To conclude the multifaceted discussion of this section, let us briefly summarize its conceptual message. We have recalled the growing interest in the literature in defining a spacetime generalization of the concept of state. In particular, this line of work has proved fruitful in the context of dS/CFT, where the spacetime states $\mathcal{R}$ can be used to provide a microscopic definition of timelike entanglement \cite{milekhin2025observable}. The relation between the SQM framework and \cite{milekhin2025observable}, together with other recent approaches, was  discussed explicitly in \cite{diaz2025spacetime}, to which we refer the reader for further details. In the present manuscript, we extend the discussion of \cite{diaz2025spacetime}, as well as the QFT version developed in \cite{diaz2023spacetime}, in order to interpret quantum-action quantization in terms of these spacetime states.
Two main results follow from this viewpoint. First, we connect the no-go theorem of section \ref{sec:DiracquantizationSCM} to the broader discussion on the need for a nontrivial spacetime generalization of the concept of state. Second, in QFT we show that the corresponding $\mathcal{R}$ for scattering processes can be written in a particularly simple form in terms of on-shell ladder operators, which arise naturally from the second-quantized QT notion of particle.

\section{Conclusions}\label{sec:concl}

In this work we asked whether the lessons of QT schemes can be lifted from single-particle relativistic models to special-relativistic QFT in a way that makes Lorentz covariance manifest at the Hilbert-space level. We showed that, while promising, a naive many-body generalization runs into a sharp obstruction tied to the notion of quantum state, and that, within the natural SCM-based line of reasoning developed here, resolving it requires abandoning Dirac-style constraint quantization in favor of a quantum-action-based quantization.

Although our analysis is carried out in the ``simple'' arena of special-relativistic QFT, it already exposes an overlooked core tension between relativity and QM: If one insists on keeping the principle of relativity explicit in the quantum formalism, then the standard notion of quantum state must be generalized. This conclusion aligns with recent results in the quantum information and quantum foundations literature, and it strongly suggests revisiting several aspects of the problem of time in this new light. In fact, since local Lorentz transformations are a special case of diffeomorphisms, the same structural issue should reappear in any generally covariant setting.

Regarding explicit covariance, let us add that we have only discussed one particular problem of it in this manuscript. Namely, we have considered field algebras that treat spacetime as a whole. This allows one to define field transformations under boosts in the same way as rotations, namely through explicit geometric unitary representations. On the other hand, the quantum action operators contain a Legendre transform part that corresponds to one particular foliation of spacetime. In particular, this means that the action does not commute with boost transformations. This is resolved by treating the foliation choice as dynamical which results in a sensible  notion of quantum reference frame \cite{giacomini2019quantum} in the QFT setting. 
This is discussed in \cite{diaz2023spacetime} showing that more than a problem of SQM it is a feature, opening interesting avenues that follow from endowing the choice of observer with a quantum aspect.  We also notice that both the spacetime symmetric field algebra discussed in this manuscript, and the foliation-related algebra of \cite{diaz2023spacetime} are particularly suited to deal with  covariance under general diffeomorphisms, thus going beyond special relativistic theories. Another aspect that remains to be explored is the interplay between the SQM formalism, a dynamical foliation, and gauge theories.

Going back to the role of QT schemes in this scenario, it could be  useful to view QT states as reduced objects of the full spacetime theory: Much like reduced density matrices provide a compressed description of many-body physics in quantum chemistry \cite{coleman1963structure}, in particular relating second quantization with single particle descriptions, QT states and related objects could be regarded as marginals of spacetime states, and thus as practical tools to simplify and tackle SQM problems. Developing the ``reduced density matrix formalism'' for spacetime states is an interesting direction for future work. From this viewpoint, QT schemes and many recent results from the QT literature \cite{boette2018history, paiva2022non,roncallo2023does,cafasso2024quantum, favalli2025relativistic, coppo2026quantum} could be particularly valuable as statements about reduced spacetime information.

Another key point to emphasize is that the SQM approach is not restricted to infinite-dimensional systems or QFT. While in this manuscript we have re-derived it from many-particle considerations, SQM is a general framework that can be applied to finite-dimensional systems as well \cite{diaz2025spacetime}. This provides a clean setting to explore the emerging notion of spacetime state and its consequences, precisely in regimes where connections with quantum information are most natural. These considerations also link QT with different lines of research in the literature regarding time in QM such as \cite{horsman2017can, fitzsimons2015quantum, cotler2018superdensity, fullwood2024operator,  milekhin2025observable, guo2025spacetime}. 
In addition, just as QT schemes have already motivated  quantum-computing protocols, such as parallel-in-time quantum computing \cite{diaz2023parallel}, one can expect that the SQM setting could provide novel computational schemes both classical, e.g. tensor network based,  and quantum.

Let us add a few comments on different approaches to QFT. Although the problem has been formulated here in a Hilbert space setting, and some connections with the PI approach were also discussed, one may naturally wonder how the present considerations relate to algebraic QFT (AQFT) \cite{haag1964algebraic}. The main point is that an algebraic counterpart of SQM would not simply reduce to standard AQFT, but would differ from it in a fundamental way. In particular, it would treat $\phi$ and $\pi$ as independent spacetime generators, leading already at the kinematical level to a notion of locality different from the standard one. We discuss this comparison in more detail in Appendix \ref{app:AQFT}.

Let us finally connect our results with the recent discussion of timelike entanglement and emergent time, in particular in the context of time as an emergent quantity under the conjectured dS/CFT correspondence \cite{harper2023timelike}. In \cite{milekhin2025observable, guo2025spacetime} the authors show that a microscopic definition of timelike entanglement, previously motivated by a ``space-to-time mapping'' \cite{harper2023timelike},  requires a genuine spacetime generalization of the notion of quantum state. In \cite{diaz2025spacetime} it is shown explicitly that this generalized notion is naturally captured by the SQM spacetime states discussed in the present manuscript.
Crucially, it is precisely the entropic properties of spacetime states, still under active development, that provide the basis for a rigorous definition of timelike entanglement beyond a simple ``space-to-time mapping''. 
This shows that the SQM formulation of QFT is not merely a convenient rewriting aimed at preserving symmetries explicitly. It is a genuinely new perspective, with concrete consequences and with the potential to sharpen our understanding of the origin of spacetime.

%%%%%%%%%%%%%%%%%%%%%%%%%%%%%%%%%%%%%%%%%%

\clearpage

\acknowledgments{I thank P. Braccia and M. Cerezo for helpful discussions and for feedback that improved the presentation of this manuscript. I also thank the discussions with Lorenzo Maccone and  R. Rossignoli on the importance of understanding the role of physical subspaces in QFT settings.
This work builds on results developed in my PhD thesis \cite{diaz2024mecanica} for which I acknowledge previous support from UNLP, IFLP, and CONICET, as well as guidance from R. Rossignoli and J. M. Matera.
This work was supported by the Center for Nonlinear Studies at Los Alamos National Laboratory (LANL) as well as by the Laboratory Directed Research and Development (LDRD) program of LANL under project number 20260043DR. }

\appendix

\section[\appendixname~\thesection]{Discrete time formalism: constraints with initial/final conditions}
We discuss here the discrete-time approach to the spacetime formalism. The comments made here should also clarify the role of the time scale $\tau$ and how initial and final conditions enter the formalism. Moreover, in order to simplify the notation and to show that the formalism is valid beyond the QFT case, we will mostly focus on the case of a single particle. 

\subsection[\appendixname~\thesubsection]{Classical case}
Let us begin by discussing the discrete spacetime version of the spacetime PBs \eqref{eq:stpbs} in $3+1$ dimensions. We have
$\{\phi_{t,x,y,z},\pi_{t',x',y',z'}\}=\delta_{tt'}\delta_{xx'}\delta_{yy'}\delta_{zz'}$  with the Dirac deltas replaced by (dimensionless) Kronecker deltas. To further ease the notation, let us replace the fields by a single scalar particle, so that $\phi_{t,x,y,z}\to q_t$, $\pi_{t,x,y,z}\to p_t$. We impose then
\begin{equation}
    \{q_t,p_{t'}\}=\delta_{tt'}\,,
\end{equation}
which replaces the standard canonical PB of classical mechanics $ \{q,p\}=1$. Let us consider time indices taking the values $t=0,1,\dots, N-1$ with $T=\epsilon N$ the total ``time window'' and spacing $\epsilon$. 
Now, if we introduce a discrete in time action
\begin{equation}\label{apeq:discreteaction}
 S= \epsilon\sum_t \left[p_t \dot{q}_t-\frac{p_t^2}{2m}-V(q_t)\right]\,,
\end{equation}
we can easily recover Hamilton equations from the constraints
\begin{align}\label{apeq:discreteconstraints}
    \{q_t, S\}&=\dot{q}_t-p_t\approx 0 \\
    \{p_t, S\}&=\dot{p}_t-V'(q_t) \approx 0
\end{align}
where we have defined a discrete derivative $\dot{q}_t:=\sum_{t',_n} \frac{i\epsilon\omega_n}{N} e^{i\omega_n (t-t')\epsilon}q_{t'}$ for $\omega_n= \frac{2\pi n}{T}$ (other possibilities, such as $\dot{q}_t:=\frac{q_{t+1}-q_t}{\epsilon}$ are possible). Notice that initial (and/or final) conditions can be imposed either in the definition of $S$ or as additional constraints. In this way, proper  discrete time versions of Hamilton equations are obtained.

\subsection[\appendixname~\thesubsection]{Quantum case}\label{app:discretequantum}
For the quantum case we impose
\begin{equation}
    [q_t,p_{t'}]=i\hbar \delta_{tt'}
\end{equation}
and quantize the action \eqref{apeq:discreteaction} accordingly. Notice that this algebra fixes the Hilbert-space representation uniquely up to unitary equivalence by the Stone-von Neumann theorem.

Interestingly, the fact that the action is quantized is yet another consequence of employing extended PBs so that operators $q_t,p_t$ at different times are independent. Let us also define $\mathcal{P}_t:=\epsilon\sum_t p_t \dot{q}_t$. Notably one can prove the following property
\begin{equation}
    e^{i\epsilon \mathcal{P}_t}q_te^{-i\epsilon \mathcal{P}_t}=q_{t+1}\,, \quad  e^{i\epsilon \mathcal{P}_t}p_te^{-i\epsilon \mathcal{P}_t}=p_{t+1}\,,
\end{equation}
namely, $\mathcal{P}_t$, the quantized Legendre transform,  generates geometrical translations in time. This simple fact can be used to prove a very powerful theorem \cite{diaz2025spacetime}: 
\begin{equation}
    {\rm Tr}\left[e^{i\mathcal{S}/\hbar}\otimes_{t=0}^{N-1} O^{(t)}_t\right]={\rm tr}\Big[e^{-iTH/\hbar}\;\hat{T}\prod_{t=0}^{N-1}O^{(t)}_H(\epsilon t)\Big]
\end{equation}
where the trace on the r.h.s. is defined in the conventional Hilbert of QM where $[q,p]=i\hbar$, $O_H(t)=e^{iH t/\hbar} O e^{-iHt/\hbar}$ are operators in the Heisenberg picture and $\hat{T}$ denotes temporal order. The intuition behind the theorem is that the map between traces can be interpreted as a generalization of the SWAP test  in quantum computing \cite{buhrman2001quantum} with $e^{i\epsilon \mathcal{P}_t}$ taking the role of a composition of SWAPS between time slices (see details in \cite{diaz2021path}).
In particular,
\begin{equation}\label{apeq:tracecorrels}
    {\rm Tr}\left[ |q\rangle_0 \langle q'| e^{i\mathcal{S}/\hbar}\otimes'_t q_t\right]=\langle q',T|\hat{T}\prod_{t}^{}{}^{'} q_H(\epsilon t)|q\rangle
\end{equation}
which is the propagator of the particle when no operator $q$ is inserted, and an $n$-point correlation function otherwise (the primer in the symbols $\otimes_t$, $\Pi_t$ denote that some subset of the $N$ time slices contain operators $q$ and identities otherwise). We see that fundamental quantities in standard QM can be recovered from timelike correlations of the operator $|q\rangle_0 \langle q'| e^{i\mathcal{S}/\hbar}$. Moreover, in \cite{diaz2021path} it was shown that when one employs the completeness relation $\int \prod_t dq_t |q_0q_1\dots q_{N-1}\rangle \langle q_0 q_1\dots q_{N-1}|=\mathbbm{1}$ to evaluate the l.h.s. of \eqref{apeq:tracecorrels} the PI of Feynman is recovered. Let us stress however that these expressions are defined independently on how we evaluate them, so that a PI-like evaluation is not necessary (one can e.g. diagonalize the action operator as discussed in \cite{diaz2021path}).

Having introduced the basic idea behind the discrete-time approach to SQM let us discuss the role of the classical constraints. At first glance, the fact that one needs to compute correlators of the exponential of the action seems completely disconnected from the constraints \eqref{apeq:discreteconstraints}. And, as discussed in the main text, these corresponds to constraints of the second class so that Dirac quantization scheme is not feasible. On the other hand, we can immediately show that
\begin{equation}\label{apeq:constraint}
    \begin{split}
         {\rm Tr}\left[ |q\rangle_0 \langle q'| e^{i\mathcal{S}/\hbar}\left( e^{i\mathcal{S}/\hbar} O_t e^{-i\mathcal{S}/\hbar} -O_t\right)\right]&={\rm Tr}\left[ |q\rangle_0 \langle q'| e^{i\mathcal{S}/\hbar}\left( e^{i\mathcal{S}/\hbar} O_t e^{-i\mathcal{S}/\hbar} -O_t\right)\right]\\
         &={\rm Tr}\left[ |q\rangle_0 \langle q'| e^{i\mathcal{S}/\hbar}\left(e^{-i \epsilon H_{t+1}/\hbar}O_{t+1}e^{i \epsilon H_{t+1}/\hbar} -O_t\right)\right]
         \\
         &=\langle q', T|\Big[e^{-i \epsilon H/\hbar}O_H((t+1)\epsilon)e^{i \epsilon H/\hbar} -O(t\epsilon)\Big]|q\rangle=0
    \end{split}
\end{equation}
which is the discrete time version of Theorem \ref{th:constraintseis} (which we proved using the weaker condition $[e^{i\mathcal{S}},\mathcal{S}]=0$, which is not satisfied when we multiply by the $|q\rangle \langle q'|$ on the initial slice). While this is the strict version we obtain for discrete time it is illustrative to rewrite the corresponding expressions for small $\epsilon$. In this case, if we use $e^{i\mathcal{S}/\hbar} O_t e^{-i\mathcal{S}/\hbar} -O_t=\frac{i}{\hbar}[\mathcal{S},O]+\mathcal{O}(\epsilon^2)$ we can rewrite \eqref{apeq:constraint} as 
\begin{equation}
    {\rm Tr}\left[ |q\rangle_0 \langle q'| e^{i\mathcal{S}/\hbar}\, [\mathcal{S},O]\right]=\langle q',T|\big(-i\hbar \dot{O}_H(t)-[H,O_H(t)]\big)|q\rangle+\mathcal{O}(\epsilon^2)\,.
\end{equation}
In other words, within brackets $\langle \dots \rangle_\mathcal{S}:={\rm Tr}[ |q\rangle_0 \langle q'| e^{i\mathcal{S}/\hbar}\dots ]$, the constraint $[\mathcal{S},O_t]$ is equal to the standard quantum expectation value of the \emph{Heisenberg equation} for the operator $O_H(t)$, and hence vanishes.

We see that just as the classical constraints $\{S,O\}\approx 0$ lead to Hamilton equations, their quantum version leads to the Heisenberg equations as desired. The only subtlety is that in order for this to hold, one need to compute physical quantities within brackets $\langle \dots \rangle_\mathcal{S}$. If we postulate this condition, then the classical formalism is immediately recovered: in the limit of $\hbar\to 0$ only a single trajectory contributes to $\langle \dots \rangle_\mathcal{S}$, and since $\langle [\mathcal{S},O]\rangle_\mathcal{S}=0$ for any value of $\hbar$ we obtain
\begin{equation}
0=\langle [\mathcal{S},O_t]\rangle_\mathcal{S}=\lim_{\hbar\to 0 }\langle [\mathcal{S},O_t]\rangle_\mathcal{S}=\{S,O_t\}|_{\text{classical-trajectory}}\,.   
\end{equation}
In conclusion, and since this holds for any $O_t$, our quantization scheme implies the constraints \eqref{apeq:discreteconstraints} in the classical limit. 

\subsection{Continuum time limit in Hilbert space, the parameter $\tau$ and regularizations}\label{app:continuumtimelimit}
In the main text we employed a continuum time formulation from the start. Here we discuss how this corresponds to considering the continuum time limit at the Hilbert space level and how it leads to the introduction of a time scale $\tau$ (similar considerations hold for the classical case).

The first step is to rescale operators $q(t):=q_t/\sqrt{\epsilon}$, $q(t):=p_t/\sqrt{\epsilon}$ such that 
\begin{equation}
   [q(t),p(t')]=\frac{1}{\epsilon}\delta_{tt'} \approx\delta(t-t')\,.
\end{equation}
By doing so, we changed the units of $q$ and $p$. This limit should be understood as the limit $\epsilon\to 0, N\to \infty$ such that $T=\epsilon N$ is fixed. Then we are still considering a time window $(0,T)$ (or $(-T/2,T/2)$). Notice also that the continuum time formulation should be understood in exactly the same sense as the usual continuum formulation of QFT. This also holds true for all dimensions: Since \eqref{eq:basicalg} is isomorphic to the standard canonical algebra in $d+1$ dimensions, no new Hilbert-space subtleties are introduced beyond those already familiar in ordinary QFT (such as the need to treat field operators as distributions, the failure of the Stone-Von Neummann theorem, etc).

As in ordinary continuum QFT, the continuum limit should be understood formally/distributionally: the fields at a point are not bona fide bounded operators, the canonical algebra does not by itself fix a unique Hilbert-space representation, and the limit from the regulated theory carries the usual UV/IR and domain subtleties. Since the algebra \eqref{eq:basicalg} is isomorphic to the standard canonical algebra in $d+1$ dimensions, these are precisely the standard continuum subtleties, not new ones specific to the present construction.

On the other hand, the continuum limit of the generator of time translations is well defined, so that
\begin{equation}
    \mathcal{P}_t \to \int dt\, p(t)\dot{q}(t)\,,
\end{equation}
satisfies
\begin{equation}
    e^{i\tau \mathcal{P}_t}q(t) e^{-i\tau \mathcal{P}_t}=q(t+\tau)\,, \quad e^{i\tau \mathcal{P}_t}p(t) e^{-i\tau \mathcal{P}_t}=p(t+\tau)
\end{equation}
for any parameter $\tau$ (with cyclic boundary conditions in time). The crucial difference with the discrete time case is that the concept of ``translating a single time step'' is no longer meaningful and a scale must be introduced to define time translations.

Moreover, for operators that are not quadratic, the continuum time limit is not straightforwardly defined. For example, $\epsilon \sum_t q_t^4=\epsilon \sum_t \epsilon^2 q^4(t)$ whose limit as an operator is not meaningful. Instead, we can replace with the well-defined  operators $\epsilon \sum_t \epsilon^2 q^4(t)\to \int dt\, \tau^2 q^4(t)$. Notably, as discussed in \cite{diaz2021path}, this parameter can be associated with a difference in the definition of the Hilbert space measure and Feynman measure in PIs. Moreover, as long as one considers the limit $\tau\to 0$ at the end of calculations, proper results are obtained (in \cite{diaz2021path} the case of finite $\tau$ is also considered).

In addition, the formulation presented in the main text corresponds to the limit $T\to \infty$. While this does not affect the operators $q(t), p(t)$, this limit corresponds to replacing discrete Fourier modes (in time), characterized by frequencies $\omega_n=\frac{2\pi n}{T}$, with continuum modes. Additional details can be found in \cite{diaz2023spacetime}. Let us however add one final remark: In this limit, the operator $e^{-iTH}$ that appears in correlation functions, becomes a projector if a small imaginary component is added to time. For this reason, the main text expressions such as \eqref{eq:aacorrel}, \eqref{eq:feynprop} correspond to vacuum correlation functions (and lead to Feynman prescription). \\

Having presented a discussion on regularizations of the spacetime algebra let us briefly describe how this translates to the breaking of Lorentz symmetry. In particular, we recall that field theories are only well-defined upon regularization, with an ensuing explicit breaking of Lorentz symmetry, as recently discussed in \cite{visser2009lorentz,polonyi2019boost,cartas2022lorentz} under the standard formalism. In the current approach, it is straightforward to introduce a regularization from the onset, without breaking the symmetric role of space and time, as it follows from time being just an index. This paves the way for a different approach to the problem that we now exemplify through the common use of a momentum cutoff. In the current scheme this corresponds to starting from the theory on a compact spacetime region $\Omega\subset \mathbb{R}^{1,d}$ such
that ladder operators in momentum space
may be expanded as 
\begin{equation}
    a(p)=\int_{\Omega} d^Dx\, e^{ip\cdot x}a(x)\,.
\end{equation}
 Notice that this implies a cutoff $\Lambda$ in momentum space, which for a compactified time leads to $\Lambda \sim 1/T$ for the $p^0$ component and $T$ the length of time.
Then, recalling that $[L_{\mu\nu}, a(x)]=(x_\mu\partial_\nu-x_\nu \partial_\mu)a(x)$, under the infinitesimal action of a Lorentz transformation we have
\begin{equation}
    [L_{\mu\nu}, a(p)]=\int_{\Omega} d^Dx\, e^{ip\cdot x}(x_\mu\partial_\nu-x_\nu \partial_\mu)a(x)=\big(p_\mu\partial_{\nu}-p_\nu\partial_{\mu}\big)a(p)
+
B_{\mu\nu}(p)
\end{equation}
where we integrated by parts with ensuing border correction
\begin{equation}
B_{\mu\nu}(p)
=
\int_{\partial\Omega} d\,\Sigma_\alpha\,
e^{ip x}\,
\big(x_\mu \delta^\alpha{}_\nu-x_\nu \delta^\alpha{}_\mu\big)\,
a(x).
\end{equation}
We see that while locally in spacetime Lorentz transformations are applied as usual, the border conditions spoil Lorentz invariance in Fourier/momentum space already at the Lie-algebra level. In the limit of a large $\Omega$ the border term contribution becomes negligible and Lorentz invariance is effectively recovered.

The formalism thus
provides a direct way to control symmetry breaking under regulators. Notice also that the previous analysis is fully kinematical and it does not involve the Hamiltonian or the action and the ladder operators are generic (not necessarily normal modes of a given theory). From this basis, one can study corrections to the quantum action under Lorentz transformations and hence quantify how much the regularization affects  observable quantities for specific theories.

\subsection{Comparison and differences with algebraic QFT and related perspectives}\label{app:AQFT}

Although the present work is formulated in a canonical Hilbert space language, it is natural to ask whether the current results are related to the standard algebraic QFT (AQFT) perspective \cite{haag1964algebraic}, even more so considering that the latter aims to provide a covariant setting that avoids the explicit use of equal-time commutation relations.  In this appendix we briefly compare the two formalisms showing that an algebraic version of the present construction substantially differs from standard AQFT in ways that further emphasize the separation between dynamical and geometrical information provided by the SQM framework.

Let us first briefly recall basic facts from AQFT in the particular example of a scalar field and under its connection to the canonical framework \cite{tjoa2022channel}.
In the algebraic approach, one assigns to each open spacetime region $\mathcal O\subset M$ a local algebra $\mathcal A(\mathcal O)$ of observables localized in that region. For a free scalar field, a basic object is the smeared field is
\begin{equation}
    \phi(f):=\int d^Dx\, f(x)\phi(x),
\end{equation}
where $f$ is a smooth compactly supported test function and $\operatorname{supp}(f)\subset \mathcal O$. Since $\phi(f)$ is typically an unbounded operator, one often works instead with the associated bounded Weyl operators $ W(f):=e^{i\phi(f)}$.
These generate the algebra $\mathcal A(\mathcal O)$ corresponding to the region $\mathcal O$, namely the $C^*$-algebra generated by all $W(f)$ with $\operatorname{supp}(f)\subset \mathcal O$. 
The Weyl operators satisfy the Weyl relations
\begin{equation}
    W(f)W(g)=e^{-\frac{i}{2}\sigma(f,g)}W(f+g),
\end{equation}
where $\sigma$ is the antisymmetric bilinear form determined by the field commutator $
    \sigma(f,g)=-i[\phi(f),\phi(g)]$. 
For the free Klein-Gordon field, one may write equivalently
\begin{equation}
    \sigma(f,g)=\int d^Dx\, d^Dy\, f(x)\,\Delta(x-y)\,g(y),
\end{equation}
where $\Delta(x-y)\equiv [\phi_H(x),\phi_H(y)]$ is a Green function of the Klein-Gordon operator and we have emphasized that in the canonical approach $\Delta(x-y)$ is related to the Heisenberg fields. This is an important remark: The bilinear form depends on the particular evolution under consideration.
Finally, we recall that in the algebraic setting the microcausality setting is
\begin{equation}
    [\mathcal A(\mathcal O_1),\mathcal A(\mathcal O_2)]=0
\end{equation}
whenever $\mathcal O_1$ and $\mathcal O_2$ are spacelike separated.

Let us now discuss the SQM counterpart of the previous construction. In the SQM approach, one starts  from the kinematical relations
\begin{equation}
[\phi(x),\phi(y)]=0,\qquad [\pi(x),\pi(y)]=0,\qquad [\phi(x),\pi(y)]=i\delta^{(D)}(x-y),
\end{equation}
where $x,y\in M$ are full spacetime points and $\delta^{(D)}(x-y)$ is the Dirac delta on spacetime. Proceeding in analogy with AQFT, for test functions $f,g\in C_c^\infty(M)$ one can define the smeared operators
\begin{equation}
\phi(f):=\int d^Dx\, f(x)\phi(x),\qquad
\pi(g):=\int d^Dx\, g(x)\pi(x)\,,
\end{equation}
implying
\begin{equation}
[\phi(f),\phi(g)]=0,\qquad [\pi(f),\pi(g)]=0, \qquad [\phi(f),\pi(g)]=i\int_M d^Dx\, f(x)g(x).
\end{equation}
In particular, if $\operatorname{supp}(f)\cap\operatorname{supp}(g)=\varnothing$, then all commutators vanish, with also $[\phi(f),\pi(g)]=0$.
A bounded version of this spacetime canonical algebra is obtained by introducing the Weyl operators
\begin{equation}
    W(f,g):=e^{i\phi(f)+i\pi(g)}.
\end{equation}
These satisfy
\begin{equation}
    W(f_1,g_1)W(f_2,g_2)
    =
    e^{-\frac{i}{2}\sigma((f_1,g_1),(f_2,g_2))}
    W(f_1+f_2,g_1+g_2),
\end{equation}
where now $\sigma$ denotes the antisymmetric bilinear form 
\begin{equation}
    \sigma\big((f_1,g_1),(f_2,g_2)\big)
    :=
    \int d^Dx\,\Big(f_1(x)g_2(x)-g_1(x)f_2(x)\Big).
\end{equation}
This formula follows directly from the smeared commutation relations above. We now remark that $\sigma$ is completely independent on the evolution, in clear contrast with standard AQFT.

Moreover, in analogy with standard AQFT, 
given an open region $\mathcal O\subset M$, one may then define the associated local algebra $\mathcal A(\mathcal O)$ as the $C^*$-algebra generated by all Weyl operators $W(f,g)$ with support on $\mathcal{O}$. Then, for all test functions supported respectively in $\mathcal O_1$ and $\mathcal O_2$ one has $\sigma\big((f_1,g_1),(f_2,g_2)\big)=0$
and therefore
\begin{equation}
    [\mathcal A(\mathcal O_1),\mathcal A(\mathcal O_2)]=0
\end{equation}
whenever $\mathcal O_1$ and $\mathcal O_2$ are disjoint. Unlike standard algebraic QFT, where commutativity is imposed only for spacelike separated regions, here the kinematical algebra leads to commutativity for arbitrary disjoint regions.

This brief comparison already shows fundamental differences between AQFT and an algebraic counterpart of the canonical approach to SQM: Firstly, since the momentum $\pi$ is completely independent of $\phi$ for all spacetime points, an algebraic approach would treat $\phi$ and $\pi$ on an equal-footing. This is in clear contrast with standard AQFT where one usually works with $\phi$ directly since $\pi$ is not an independent generator as it can be formally recovered from $\phi$ alone.
Secondly, no dynamical information enters the basic setup that defines the algebras (and in particular bilinear forms). This is the algebraic version of the fact that SQM encodes evolution itself on the action/spacetime states. 
From these basic considerations one may develop a version of AQFT where evolution is recovered a posteriori lifting the results of section \ref{sec:quantumactionquantization}. 
In fact, the previous considerations suggest, at least in principle, a generalization of the standard AQFT notion of states as linear functionals, which in a Hilbert space representation takes the form ${\rm Tr}[\rho\,\dots]$, to linear functionals corresponding to quantities of the form ${\rm Tr}[\mathcal{R}\,\dots]$, properly taking into account the features of $\mathcal{R}$ discussed in section \ref{sec:ststates}.

\section[\appendixname~\thesection]{Spacetime states of transitions through on-shell modes}\label{app:Smatrix}

In this appendix we discuss how one can define a spacetime state corresponding to the $S$-matrix transitions. Our discussion assumes continuum and unbounded time as in the main text.

Let us start by noting that for $-T<t<T$ (here $T$ is an arbitrarily large positive time) 
\begin{equation}
    \frac{{\rm Tr}[e^{i\mathcal{S}_\tau}C(-T) O(t) A(T)]}{{\rm Tr}[e^{i \mathcal{S}_\tau}]}\equiv \langle 0| A_H(T) O_H(t) C_H(-T)|0\rangle
\end{equation}
for $|0\rangle$ the ground state of the Hamiltonian corresponding to $\mathcal{S}_\tau$ and where we omit the explicit $\tau$ limits for ease of notation.
In particular, if $A$ is a product of annihilation operators and $C$ of creation operators it is clear that we are describing excited states. The corresponding spacetime state (or transition) is hence given  by $A(T) e^{i\mathcal{S}_\tau}C(-T)$ as it follows from the cyclicity of the trace, and in agreement with the main text discussion.

Consider now the case $O\equiv \mathbbm{1}$ so that we are only describing a transition. Let us ``extract'' one of the creation operators so that $C(-T)= C'(-T) a^\dag(-T,\textbf{p})$ and replace it with $C'(-T) a^\dag(-T,\textbf{p})\to C'(-T) \left(a^\dag(-T,\textbf{p})+a^\dag(T,\textbf{p})\right)$. Now the previous equation reads
\begin{equation}
\begin{split}
       \frac{{\rm Tr}[e^{i\mathcal{S}_\tau}C'(-T)\left(a^\dag(-T,\textbf{p})+a^\dag(T,\textbf{p})\right)A(T)]}{{\rm Tr}[e^{i \mathcal{S}_\tau}]}&\equiv \langle 0| A_H(T) C'_H(-T)a^\dag_H(-T,\textbf{p})|0\rangle\\&+\langle 0| A_H(T) a^\dag_H(T,\textbf{p}) C'_H(-T)|0\rangle\,.
\end{split}
\end{equation}
If none of the out particles has the momentum of the in particles the second term vanishes. Instead, under standard assumption, for large $T$ we expect $a^\dag_H(-T,\textbf{p})|0\rangle\simeq e^{-i E_{\textbf{p}}T}a^\dag(\textbf{p})|0\rangle$. This means that we can safely replace the creation operator with the difference between two times, and we can get rid of the phase by considering $C(-T)\to C'(-T) \left( e^{i E_{\textbf{p}}T}a^\dag(-T,\textbf{p})- e^{-i E_{\textbf{p}}T}a^\dag(T,\textbf{p})\right)$. Let us recall that a very similar procedure is usually employ to derive the LSZ reduction formula \cite{maggiore2005modern}. We can now make use of the asymptotic relation \cite{weinberg2005quantum} ($T\gg 1$)
$$
f(T)+f(-T)=\lim_{\epsilon\to 0^+}\epsilon \int_{-\infty}^\infty dt\, f(t) e^{-\epsilon |t|}
$$
to write
\begin{equation}
    \left(e^{-i E_{\textbf{p}}T}a^\dag(T,\textbf{p})+e^{i E_{\textbf{p}}T}a^\dag(-T,\textbf{p})\right)=\lim_{\epsilon\to 0^+}\epsilon \int_{-\infty}^\infty dt\, a^\dag(t,\textbf{p}) e^{-i E_{\textbf{p}} t}e^{-\epsilon |t|}\approx -i (p^0-E_{\textbf{p}}+i \epsilon) a^\dag(E_{\textbf{p}},\textbf{p})
\end{equation}
where we used the Fourier transform definition of the extended momenta modes. Notice that this last step is only accessible through the spacetime formalism. Combining all these results we have proven that we can replace 
\begin{equation}
    A(T) e^{i\mathcal{S}_\tau}C(-T)\to -i (p^0-E_{\textbf{p}}+i \epsilon)A(T) e^{i\mathcal{S}_\tau}C'(-T)a^\dag(E_{\textbf{p}},\textbf{p})
\end{equation}
without affecting the trace.

We can do a completely analogous discussion with the remaining creation and annihilation operators. In addition, let us recall that the external states that define particles in the standard $S$ matrix formalism are given by $|\textbf{p}\rangle\equiv \sqrt{2E_{\textbf{p}}}a^\dag(\textbf{p})|0\rangle$. In order to match this convention we need an additional factor for each ladder operator so we end up with 
\begin{equation}
\begin{split}
      A(T) e^{i\mathcal{S}_\tau}C(-T)&\to \mathcal{R}':=\prod_{i=1}^n [-i \sqrt{2E_{\textbf{p}_i}\tau}(p_i^0-E_{\textbf{p}_i}+i \epsilon)] \prod_{j=1}^{n'} [-i \sqrt{2E_{\textbf{k}_j}\tau}(k_j^0-E_{\textbf{k}_j}+i \epsilon)]\times \\ &\times a(E_{\textbf{k}_1},\textbf{k}_1)a(E_{\textbf{k}_2},\textbf{k}_2)\dots a(E_{\textbf{k}_{n'}},\textbf{k}_{n'}) e^{i\mathcal{S}_\tau} a^\dag(E_{\textbf{p}_1},\textbf{p}_1)a^\dag(E_{\textbf{p}_2},\textbf{p}_2)\dots a^\dag(E_{\textbf{p}_n},\textbf{p}_n)\,,
\end{split}
\end{equation}
where we have also reintroduced the factors $\sqrt{\tau}$, one for each ladder operator. This is the final formula that leads to ${\rm Tr}[\mathcal{R}']$ being equal to the $S$ matrix elements \footnote{We implicitly omitted a small detail: The presence of a field renormalization strength $Z$ which maybe reintroduced by replacing $\tau\to Z \tau$ and corresponds to a proper rescaling of the asymptotic ladder operators that define asymptotically free states up to a rescaling.  }.

Considering the many assumption involved in the previous derivation let us add a few comments on the formula and explicitly show how it
is applied in practice, thus confirming its correctness. For concreteness let us consider the first order  of a $2\to 2$ scattering process in the $\lambda \phi^4$ model. We have
\begin{equation}
\begin{split}
      {\rm Tr}[\mathcal{R}']= \prod_{i=1}^2 [-i \sqrt{2E_{\textbf{p}_i}\tau}(p_i^0-E_{\textbf{p}_i}+i \epsilon)] \prod_{j=1}^{2} [-i \sqrt{2E_{\textbf{k}_j}\tau}(k_j^0-E_{\textbf{k}_j}+i \epsilon)]\times \\
     \times (-i \frac{\lambda}{4!} \tau^2) \int d^Dz\, \langle \phi^4(z)  a^\dag(E_{\textbf{p}_1},\textbf{p}_1) a^\dag(E_{\textbf{p}_2},\textbf{p}_2)a(E_{\textbf{k}_1},\textbf{k}_1)a(E_{\textbf{k}_2},\textbf{k}_2)\rangle_\tau+\mathcal{O}(\lambda^2)
\end{split}
\end{equation}
where the brackets indicate expectation value with respect to the free action. This means that we can use Wick's theorem and decompose the mean value in product of two-point contractions. The relevant contractions are 
\begin{equation}
    \begin{split}
        \langle \phi(z)  a^\dag(E_{\textbf{p}},\textbf{p})\rangle_\tau&= \frac{e^{-i pz}}{\tau \sqrt{2 E_{\textbf{p}}}}\frac{i}{p^0-E_{\textbf{p}}+i\epsilon}+\mathcal{O}(\tau)\\
               \langle \phi(z)  a(E_{\textbf{k}},\textbf{k})\rangle_\tau&= \frac{e^{i kz}}{\tau \sqrt{2 E_{\textbf{k}}}}\frac{i}{k^0-E_{\textbf{k}}+i\epsilon}+\mathcal{O}(\tau)\,.
    \end{split}
\end{equation}
We see that the overall factors precisely cancel the poles arising as a consequence of these contractions. Just as when working with the LSZ reduction formula, one  comput equantities off-shell, cancel the pole factors and then consider the on-shell limit. Notice that within the spacetime formalism working off-shell can be actually made rigorously since the extended ladder operators are well-defined for arbitrary momenta. One may then employ the intuition developed in \cite{diaz2019historystate, diaz2019history} for single particles where a time operator leads to a mass operator. Returning to the previous equation we obtain to first order in $\tau$ and $\lambda$
\begin{equation}
\begin{split}
      {\rm Tr}[\mathcal{R}']\equiv  
    -i \lambda  \int d^Dz\, e^{-iz (p_1+p_2-k_1-k_2)}= -i \lambda (2\pi)^D \delta^{(D)}(p_1+p_2-k_1-k_2)
\end{split}
\end{equation}
which is the standard result with the delta indicating momenta conservation. The considerations made within the example apply to general scattering processes.

\section[\appendixname~\thesection]{Dirac Fermions}\label{app:fermions}

The result we presented for bosons in the main text can be easily generalized to fermions, with similar considerations applying to both. In particular, a single particle QT approach to Dirac equation has been developed in \cite{diaz2019historystate}. Therein it is shown that the corresponding single particle basis is given by states $|x,\sigma\rangle\equiv |t\rangle \otimes |\textbf{x}\rangle \otimes |\sigma\rangle$, namely we have an additional spinorial sector with respect to bosons, the usual spatial part, and an additional basis of time-eigenstates.
Similarly to what we discussed in section \ref{sec:sp}, one introduces a universe equation $(p^0+H_D)|\Psi\rangle=0$ now leading to Dirac's equation upon projection onto $\langle x|$ and for $H_D$ the Dirac hamiltonian. An equivalent condition is $\gamma^\mu p_\mu|\Psi\rangle=m|\Psi\rangle$ highlighting that physical states are mass eigenstates.

When generalizing this scheme to the second quantization setting, one identifies
 \begin{equation}
     |x,\sigma\rangle\equiv \psi^\dag_\sigma(x)|\Omega\rangle\,,
 \end{equation}
and imposes an antisymmetric algebra corresponding to 
\begin{equation}\label{eq:anticommfields}
    \{\psi_a(x),\psi^\dag_b(y)\}=\delta^{(4)}(x-y)\delta_{ab}\,,
\end{equation}
and to $\mathcal{H}=\text{Antisym}[\oplus_n h^{\otimes n}]$. This forces us to use anti-commutation relations both for space and time. Interestingly, this means that there is no longer a tensor product structure across time
as in the bosonic case. In particular, this means that the SWAP related discussion we presented in \ref{sec:quantumactionquantization}, with the quantum action involving a concatenation of SWAPs, cannot hold. Nonetheless, if we consider the second quantization of the universe operator one finds (the steps are the same as in the bosonic case; see \cite{diaz2025spacetime} for an explicit proof)
\begin{subequations}\label{eq:dqaction}
\begin{align}
    \mathcal{S}
    = \int dt\, \left[ \int d^dx\,\psi^\dag(t,\textbf{x})i\dot{\psi}(t,\textbf{x})- H_D(t)\right]
   =\int \!d^Dx\, \bar{\psi}(x)(\gamma^\mu i\partial_\mu-m){\psi}(x)\,,  \label{eq:relaction}
\end{align}
\end{subequations}
for $H_D(t)=\int d^3x\, \psi^\dag(t,\textbf{x})(-i\boldsymbol{\alpha}\cdot \nabla+\beta m)\psi(t,\textbf{x})$ the Dirac Hamiltonian as a function of spacetime operators, at a given time $t$. We see that the fermionic quantum action also emerges naturally from the second quantized universe operator of the correpsonding QT scheme.

Notably, 
\begin{equation}
\begin{split}
     \mathcal{P}
    =\int dt\int d^3x\, \psi^\dag(t,\textbf{x})i\dot{\psi}(t,\textbf{x})=\int d^4x\, \bar{\psi}(x)\gamma^0 i\partial_0{\psi}(x)\,,
\end{split}
\end{equation}
which has the form of the classical \emph{Legendre transform} for a Dirac field, generates time translations: 
\begin{equation}
    e^{i\tau \mathcal{P}}\psi(t,\textbf{x}) e^{-i\tau \mathcal{P}}=\psi(t+\tau,\textbf{x})\,,
\end{equation}
just as in the bosonic case even if this relation follows now from anticommutators. As a matter of fact, one can develop a discrete time version of the formalism \cite{diaz2025spacetime} making use of a variant of the operator fSWAP, namely the fermionic version of the SWAP operator that properly takes into account the antisymmetric nature of fermions: For $N=2$, and a single fermionic mode 
$e^{i\epsilon\mathcal{P}}=
\left(\begin{smallmatrix}
1&0&0&0\\
0&0&1&0\\
0&-1&0&0\\
0&0&0&1
\end{smallmatrix}\right)
$
which has an additional minus sign with respect to a standard SWAP. For general $N$ the operator $e^{i\epsilon \mathcal{P}}$ is a composition of these operators across different time modes. 
It is also important to remark that fSWAP is a gaussian operator, namely the exponential of a quadratic fermionic operator, while SWAP is not (in fact free fermionic operations together with SWAP lead to universal unitaries \cite{hebenstreit2019all}).

Now, by using the PBs version of Eq.\ \eqref{eq:anticommfields} it is easy to see that the SCM formalism we discussed in the main text holds for fermions. In particular, the constraint action $\{S,\psi^\dag\}\approx 0$ leads to the Dirac (field) equation. These constraints are also of the second kind yielding the same kind of problems we discussed in section \ref{sec:limitations}. Just as in the bosonic case the right approach is to exponentiate the action and consider expectation values of insertion of operators (with an additional parity operator $P$ to be included for fermions; this turns out to be to the fermionic path integral): If we define $
    \langle \dots \rangle_\tau:=\frac{{\rm Tr}[P e^{i\tau\mathcal{S}}\dots]}{{\rm Tr}[P e^{i\tau\mathcal{S}}]}\,,$ and, in addition,
    we consider the FT of the fields, so that the action reads
  $
    \mathcal{S}=\int \frac{d^Dp}{(2\pi)^4}\psi^\dag(p)\gamma^0(\gamma^\mu p_\mu-m)\psi(p)\,,
$ elementary fermionic properties yield 
\begin{align}
    \langle \psi(p)\bar{\psi}(k)\rangle_\tau&=\frac{\delta^{(D)}(p-k)}{\mathbbm{1}-\exp[i\tau \gamma^0(\gamma^\mu p_\mu-m)/(2\pi)^D]}\gamma^0\label{eq:pprop}\\
    &=\frac{1}{\tau}\frac{i}{\gamma^\mu p_\mu-m}(2\pi)^4\delta^{(D)}(p-k)+\mathcal{O}(\tau)\,,\label{eq:ppropstau}
\end{align}
where we used the basic property $\gamma^0 \gamma^0=1$ leading to $[\gamma^0 (\gamma^\mu p_\mu-m)]^{-1}=(\gamma^\mu p_\mu-m)^{-1}\gamma^0$.
We see that for small $\tau$ the FT of the Dirac propagator is obtained, with \eqref{eq:pprop} an analytic function for any real $\tau$ taking $m^2\equiv m^2-i\epsilon$. As a result,
\begin{equation}\label{eq:diracprop}
     \lim_{\tau\to 0} \langle \sqrt{\tau}\psi(x)\sqrt{\tau}\bar{\psi}(y)\rangle_\tau=\!\int \!\!\frac{d^Dp}{(2\pi)^D}\frac{i(\gamma^\mu p_\mu+m)}{p^2-m^2+i\epsilon}\,e^{-ip(x-y)}
\end{equation}
which is just the position space propagator of a Dirac field.

We have schematically shown how all the main text discussion about bosons holds for fermions. Along these lines we can also develop fermion-boson interacting field theories. Let us also mention that as usual, quantizing 
Dirac fields with commutators leads to negative energies so that the proper statistic are indeed required. We refer the reader to \cite{diaz2025spacetime} for additional details on the relation between SQM and a QT approach to fermions, including a discussion on the notion of antiparticle when time is quantum.

\end{document}